\documentclass[12pt]{article}

\usepackage{verbatim}
\usepackage{graphicx}
\usepackage{epsfig}
\usepackage{psfrag}
\usepackage{wrapfig}
\usepackage[all]{xy}

\usepackage{geometry}
\usepackage{setspace}
\usepackage{url}
\usepackage{color, soul}
\usepackage{enumerate}
\usepackage{booktabs}
\usepackage{longtable}
\usepackage{todonotes}

\usepackage{footmisc}

\newcommand\wordcount{
    \immediate\write18{texcount -sum -1 \jobname.tex > 'count.txt'}
\input{count.txt}words}

\usepackage{natbib}

\usepackage{amsmath}
\usepackage{amsfonts}
\usepackage{amsthm}
\usepackage{bbm}
\usepackage{array}

\newtheorem{prop}{Proposition}

\def\Var{\mathbb{V}\,}
\def\E{\mathbb{E}\,}

\def\Supp{{\rm Supp}\,}

\newcommand\independent{\protect\mathpalette{\protect\independenT}{\perp}}
\def\independenT#1#2{\mathrel{\rlap{$#1#2$}\mkern2mu{#1#2}}}

\begin{document}

\sloppy
\author{Michael Gechter\thanks{Michael Gechter is Assistant Professor, Department of Economics, The Pennsylvania State University, Kern Building, University Park PA 16802 (Email: mdg5396@psu.edu).}  \\ Cyrus Samii \thanks{Cyrus Samii is Associate Professor, Department of Politics, New York University, 19 West 4th St., New York, NY 10012 (Email: cds2083@nyu.edu).}\\  Rajeev Dehejia\thanks{Rajeev Dehejia is Professor, Robert F. Wagner Graduate School of Public Service, New York University, 295 Lafayette St., New York, NY 10012 (Email: dehejia@nyu.edu)} \\ Cristian Pop-Eleches\thanks{Cristian Pop-Eleches is Associate Professor, School of International and Public Affairs, Columbia University, 420 West 118th St., New York, NY 10027 (Email: cp2124@columbia.edu).}}

\title{Evaluating {\it Ex Ante} Counterfactual Predictions Using {\it Ex Post} Causal Inference\footnote{We are grateful to David Atkin, Abhijit Banerjee, Sylvain Chassang, Andrei Gomberg, Kei Hirano, Hiroaki Kaido, Rohit Lamba, Rachael Meager, Minchul Shin, Jann Spiess, and seminar/workshop participants at Bristol, the Canadian Econometric Study Group, Cambridge, EPED, J-PAL, Oxford, Penn State, Rice, the Tsinghua Econometrics Conference, UCL, and the Y-RISE Inaugural Conference for helpful comments. We are indebted to Ben Olken and Gabriel Kreindler for sharing code and data from \cite{Banerjee2015b}. Shayleen Reynolds, Sanjna Shenoy, Chung Han Yang, Kaida Zhang, and Junlong Aaron Zhou provided excellent research assistance.}}
\date{\today
	\\v0.3}

\maketitle

\begin{abstract} 
We derive a formal, decision-based method for comparing the performance of counterfactual treatment regime predictions using 
the results from randomized experiments. 
Our approach allows us to quantify and assess the statistical significance of differential performance for optimal treatment regimes estimated from structural models, extrapolated treatment effects, expert opinion, and other methods.
We apply our method to evaluate optimal treatment regimes for conditional cash transfer programs across countries where predictions are generated using data from experimental evaluations in other countries and pre-program data in the country of interest.
\end{abstract}


%
\clearpage
\doublespacing
\section{Introduction}

Social scientists often provide recommendations about the implementation of policies, which determine whether and in what manner a given treatment should be applied in some target context. 
Crucial to this task is generating counterfactual predictions to inform such recommendations.
The methods to which social scientists turn for this purpose are varied.
They include quantitative extrapolations from existing randomized evaluations or observational studies, predictions based on structural models that interpret behavioral patterns in the target context, and more subjective expert opinions, among others.

Our goal in this paper is to offer a formal framework for evaluating the relative success of methods for generating policy recommendations.
We evaluate the success of {\it ex ante} policy recommendations that can draw on pre-existing experiments and descriptive data to recommend a treatment assignment in a target context. 
Then, we conduct an {\it ex post} evaluation of the recommendation, based on information from  experiments in the target context. 
We are thus able to quantify and assess the statistical significance, {\it ex post}, of the differential performance of various methods for generating {\it ex ante} recommendations.
These {\it ex post} assessments can inform choices over methods to produce {\it ex ante} recommendations for new contexts where experimental data are not yet available.

Our framework is decision-based, considering optimal choices for a social planner seeking to maximize a  welfare objective subject.
We consider a linear social welfare function, although the framework could take into consideration preferences over inequality, uncertainty, status quo bias, or other modifications.
The planner is thought to face constraints that limit the set of policy options.
Multiple methods may yield identical results for basic types of recommendations --- e.g., they may agree on whether the treatment is harmful or helpful on average, in which case they would make the same recommendation regarding the alternatives of ``treat everyone'' versus ``treat no one.''

But when it comes to estimating optimal policies, methods may differ in their recommendations.
Such differences in recommendations would be due to different methods' predictions about how different types of people respond to treatment  \citep{Manski2004, Dehejia2005, Hirano2009, Imai2011, Kitagawa2015, Athey2017a} .
The wider the range of admissible policy options, the wider is the scope for methods to differ in their recommendations and thus in their relative performance.
Another way to put this is that the more refined the policy decision at stake, the more we need to ask of the methods, and thus the more refined will be our judgment about the performance of different methods.

What we refer to as a ``method'' is an approach for determining which of these policy options should be implemented.
These include reduced form methods that rely on conditional unconfoundedness to extrapolate conditional treatment effects using existing experimental or observational evaluations from other contexts, as in  \cite{JosephHotz2005} and \cite{Dehejia2014}.
Another class of methods includes meta-analytical methods that model heterogeneity across contexts, as in \cite{Dehejia2003},  \cite{Meager2015}, and \cite{Vivalt2014}.
A third class of methods includes structural models.
These include models that interpret behavioral patterns among untreated observations within the target context so as to generate counterfactual predictions about outcomes under treatment, as in \cite{Todd2006} and \cite{Todd2010a}.
They also include approaches that estimate of some structural model parameters using untreated units in the target context and others from available experimental evidence, as in \cite{Attanasio2012}.
Hybrids of reduced form and structural methods are also available, as in  \cite{Gechter2015a}.\footnote{Our framework as currently developed does not allow for interval predictions for treatment effects, like those generated by \cite{Gechter2015a} and \cite{Andrews2017} but integrating them through the theory of treatment choice under ambiguity reviewed in \cite{Manski2011} presents no particular problem.}
Finally, a fifth class of methods includes the solicitation of subjective opinions of experts, as discussed in \cite{Banerjee2016b} and applied in \cite{DellaVigna2016a}.

In our application, we allow methods to draw on existing experimental data from other contexts as well as data on covariates and untreated outcomes in the target context.
The resulting recommendations are thus ``{\it ex ante}'' insofar as they are made under uncertainty about the distribution of potential outcomes under treatment.
An assessment is then done {\it ex post}, using data from experiments that reveal the distribution of potential outcomes under treatment.
Our analysis allows us to determine, {\it ex post}, whether the policy recommendation from one or another method performed significantly better in terms of promoting welfare.   

We use our framework to evaluate conditional cash transfer policies for increasing children's school enrollment, a policy subject to widespread consideration around the world \citep{Parker2018}.
We begin with an illustrative example of making a policy recommendation for a conditional cash transfer program in Morocco.  {\it Ex ante}, the methods can draw upon data from a conditional cash transfer experiment in Mexico---specifically, the PROGRESA randomized evaluation---as well as covariates and untreated outcome data in Morocco.
We consider four methods:
(1) simple reduced form extrapolation of conditional treatment effects by age-sex strata from Mexico to Morocco, (2) more sophisticated conditional average treatment effect extrapolation using \cite{Athey2019-grf}'s Generalized Random Forest (GRF) algorithm, (3) a static, semi-parametric structural (SPS) model based on \cite{Todd2010a}, (4) a dynamic, parametric structural model based on \cite{Attanasio2012} (AMS).
We then assess, {\it ex post}, the performance of the two methods using the results of the randomized evaluation conducted by  \cite{Benhassine2015}. 

The results show
that methods (1) and (4) perform best in this particular case despite making fairly different policy recommendations.   
The GRF appears to suffer from ``contextual overfitting'' in the sense that the extra heterogeneity it detects is more idiosyncratic to PROGRESA than the basic heterogeneity captured by a simple stratification exercise.
SPS's non-parametric component is perhaps too flexible, leading to a more standard overfitting problem despite some regularization.  
These results are preliminary, based on one reference-target pair out of 7 different CCT contexts we have harmonized data from.
Our full set of findings are intended to speak to the question of whether policy recommendations should rely on internally valid evidence generated outside the target context, or whether one should rely on potentially confounded evidence from within the target context, as in \cite{Pritchett2013}.

Our methodological contribution is a formal framework for evaluating policy recommendations based on counterfactual predictions from competing methods.
Our aim is to provide tools that are more general in speaking to policy recommendations than the relatively informal and case-specific model validation exercises that regularly appear in the applied microeconomics literature; see, for example, \cite{Todd2006}, \cite{Keane2007a},  \cite{Duflo2012}, and \cite{Wolpin2013}, who each use predictions into holdout samples to evaluate the fit of structural models.
Our framework for evaluating policy recommendation methods builds on \cite{Pesaran2002} and \cite{Granger2006}'s notion of using information on the intended use of predictions as a basis for judging methods, foundational work on forecast evaluation theory by \cite{Diebold1995}, \cite{White2003},  and \cite{Hansen2011}, as well as work on prediction-based model comparisons by \cite{Keane2007a}, \cite{Wolpin2007}, \cite{Wolpin2013}, \cite{Schorfheide2012}, and \cite{Schorfheide2016}.
We see our work as complementary to research by  \cite{DellaVigna2016a} on evaluating the quality of experts' {\it ex ante} forecasts by using experimental data {\it ex post}.
\cite{Manski1996, Manski2004}, \cite{Dehejia2005}, \cite{Hirano2009}, \cite{Tetenov2012}, \cite{Kitagawa2015}, and \cite{Athey2017a} address the issue of deriving optimal treatment regimes in decision-theoretic terms in sample; we extend these ideas to the evaluation of methods out of sample as well.

Evaluating methods through the lens of a decision problem is in line with \citet{Arrow1972}, and presents an attractive alternative to model comparison based on Kullback-Leibler divergence (as in \citealp{Vuong1989}). 
Selection based on Kullback-Leibler superiority may be inadequate for a particular decision problem \citep{Marschak1968} .
Finally, our application to the conditional cash transfer programs draws on the synthesis discussions in \cite{Banerjee2015b} and \citet{Garcia-Saavedra2017-CCT-meta}
as well as the specific data and analyses conducted by \cite{DeJanvry2006}, \cite{Todd2006}, and \cite{Attanasio2012} for Mexico and \cite{Benhassine2015} for Morocco.
	
\section{Setting}

Let $\mathcal{M}$ define a set of methods under consideration.  
A method $m \in \mathcal{M}$ produces predictions for a set of real-valued treatment conditions, $\mathcal{T}$, finite.  
Suppose that the status quo treatment condition is given by $t= 0 \in \mathcal{T}$ for all $m \in \mathcal{M}$.  
Considering our application below, we can imagine that $m$ could be a reduced form extrapolation method for predicting responses to different subsidy schedules, which are elements in the set $\mathcal{T}$, including the no subsidy condition, $t=0$.  

We consider a social planner interested in using data from a set of reference contexts to devise welfare-maximizing policies in a set of target contexts.  
Suppose contexts are indexed by $c \in \{1,...,C\}$.  
Let $D_c$ be an indicator variable dividing the contexts into target and reference contexts, such that $D_c = 1$ when $c$ is a target context and $D_c = 0$ when $c$ is a reference context. 
The planner wants to assess the methods in $\mathcal{M}$ according to their ability to assign individuals to treatments from within $\mathcal{T}$ in a way that maximizes the planner's social welfare function in the target contexts, subject to potential constraints on feasible treatments.   

A set of treatment conditions, $\mathcal{T}_c \subseteq \mathcal{T}$, is active in each context $c$, although methods may be able to use data from a set $\mathcal{T}_c$ to generate counterfactual predictions for treatments that are outside this set.  For example, structural methods can generate counterfactual predictions for treatment effects with treatments that have never been implemented. 
For the target contexts, we suppose that {\it ex ante}, $\mathcal{T}_c = \{0\}$, meaning only data on the status quo conditions are available.  
In the reference contexts, {\it ex ante}, $\mathcal{T}_c$ may contain treatments other than just the status quo.  
Thus, in the target contexts, only the $t=0$ treatment is active prior to the implementation of policy, and the social planner is seeking a  recommendation on how to introduce treatments from a feasible set of options so as to maximize welfare.  
The recommendations are based on methods that make counterfactual predictions for the target contexts.  
We observe an $J$-vector of context level characteristics, $V_c$, which contains the indicator, $D_c$, as well as $J-1$ other attributes, such that $\Supp(V_c) = \mathcal{V}\subseteq \mathbb{R}^J$.  

Within a given context $c$, let $i$ index individuals. 
Each context is governed by a probability distribution, $\mathbb{P}_c$, on the following: an individual-level treatment variable, $T_{ic}$, where $\Supp{(T_{ic})} = \mathcal{T}_c$; a $K$-vector of individual-level covariates, $W_{ic}$, with  $\Supp{(W_{ic})} = \mathcal{W}_c \subseteq \mathbb{R}^K$; and an individual-level potential outcome function, $Y_{tic}:\mathcal{T}_c \rightarrow \mathbb{R}$, that maps treatments to outcomes.  
These outcomes measure individual wellbeing from the planner's perspective.
We suppose the following conditions on the outcome data:
\begin{description}
\item[C0.] $| \E_c[Y_{tic}] | < \infty$ and $0 < \Var_c[Y_{tic}] < \infty$ for all $ t \in \mathcal{T}_c$ and $c \in \{1,...,C\}$.  
\end{description}
Let $Y_{it} = Y_{T_{ic}ic}$ be the observed outcome for unit $i$ in context $c$.  
The observed data consist of the context level characteristics and individual-level data, $O_{ic} = (V_c, W_{ic}, T_{ic}, Y_{ic})$, for random samples of individuals $i \in \{1,...,n_c\}$ across contexts $c=\{1,...,C\}$.  

For simplicity, we focus on the case where we can take individual-level treatments to be binary, in which case $\mathcal{T} = \{0,1\}$, where $t=0$ is a status quo control condition and $t=1$ a treated condition.  
Then, individuals' potential outcomes are $\left(Y_{1ic},  Y_{0ic}\right)$ and observed outcomes are given by $Y_{ic} = T_{ic}Y_{1ic}, + (1-T_{ic})Y_{0ic}$. 
For a context $c \in \{1,..,C\}$, consider the following conditions on the data generating process defining $\mathbb{P}_c$: 
\begin{description}
\item[C1.] conditional random assignment, $\left(Y_{1ic},  Y_{0ic} \right) \independent T_{ic}| W_{ic} = w$ for all $w \in \mathcal{W}_c$; 
\item[C2.] overlap, where probabilities of assignment, $p_c(w) := \Pr[T_{ic} = 1|W_{ic} = w]$, are interior such that $0 < b_0 < p_c(w) < b_1 < 1$ for all $w \in \mathcal{W}_c$, and
\item[C3.] stable unit treatment value assumption (SUTVA; \citealp{rubin80_randomization_analysis}), such that for $T_{ic} = t$, $Y_{ic} = Y_{tic}$ with probability 1.
\end{description}
When conditions C1-C3 hold, the conditional (on covariates) potential outcome distributions for both the treatment and control groups in the reference contexts are identified.  

In our setting, we suppose that, {\it ex ante}, we are working with a random sample for which conditions C1-C3 hold in the reference contexts (for which $D_{ic} = 0$).
However, in the the target contexts (for which $D_c = 1$), we suppose that, {\it ex ante}, we have a random sample only for units with $T_{ic} = 0$, while the covariate distributions are also identified.   
As such, {\it ex ante}, the distributions $\mathbb{P}_c[(Y_{0ic}, W_{ic})|V_c, D_c = 0]$, $\mathbb{P}_c[(Y_{1ic}, W_{ic})|V_c, D_c = 0]$, and $\mathbb{P}_c[(Y_{0ic}(0), W_{ic})|V_c, D_c = 1]$ are identified.  
The methods make use of these distributions to derive policy recommendations.  
Once the recommendations are submitted, we then suppose that the conditions C1-C3 obtain in the target contexts, in which case the $\mathbb{P}_c[(Y_{1ic}, W_{ic})|V_c, D_c = 1]$ distributions are revealed.  This allows the planner to judge, {\it ex post}, the quality of the methods in terms of how their recommendations fair with respect to welfare.

Conditions C1-C3 include situations where we have a set of randomized experiments that we can feed into a set of methods to produce {\it ex ante} recommendations for a new target context.
These conditions also admit observational studies in which conditional independence and overlap of treatment assignment holds over the covariate space, $\mathcal{W}_c$, although for the time being we focus on randomized experiments.
Our specification of the potential outcomes, combined with assumption C3, rules out ``interference'' (\citealp[p. 19]{cox58}), including general equilibrium effects. 
We consider this a starting point for our analysis. 
Generalizations could consider situations where interference is present, and derive criteria for judging methods by working with causal quantities that are identified under interference \citep{hudgens_halloran08, aronow_samii2017_interference}.

\section{Planner's Objective and {\it Ex Ante} Recommendations}

We can define the planner's objective in context $c$ in general terms as
$$
\max_{\pi} \mathcal{U}(\mathbb{P}_c(\pi))
$$ 
where $\pi : \mathcal{V} \times \mathcal{W}_c \rightarrow \{0,1\}$ is a treatment assignment rule that takes  context $c$'s contextual attributes and an individual's covariate values and maps them to a treatment assignment.
Then,
$\mathbb{P}_c(\pi)$
is the 
joint potential outcome-covariate distribution
induced by $\pi$.
The restriction to  $\{0,1\}$ for each individual in the target context follows from our focus on binary treatments.  
If a richer set of treatment values $\mathcal{T}$ were under consideration, the rule $\pi$ could assign distribution functions over this $\mathcal{T}$ on the basis of covariate values.  
Sometimes the set of conditional treatment assignments that maximizes this objective {\it ex ante} is non-unique---i.e., there are ties. 
For example, multiple units may share the same covariate value.  
The treatment assignment that maximizes the objective may assign some fraction of such units to treatment.
Then, all permutations of assignments would yield the same {\it ex ante} value for the objective. 
We assume that $\pi$ encodes a tie-breaker rule that is unconditionally statistically independent and equalizes probability of treatment for such tied units.

Current approaches to this problem (as in \citealp{Manski2004} and \citealp{Tetenov2012}) consider social welfare that is linear in expected treatment and control outcomes, what \citet{Hirano2019} refer to as utilitarian welfare.
We adopt the same approach.\footnote{For the current analysis, we put to the side considerations related to the planner's sensitivity to inequality or  uncertainty \citep{Dehejia2008}, as well as asymmetry in the planner's preferences toward different treatments, such as status quo bias \citep{Tetenov2012}.}
We also incorporate treatment opportunity costs by adding a cost-effectiveness term to the planner's objective.
This increases the planner's sensitivity to over- and under-estimation of the treatment effects in generating treatment rules. 
We suppose that we are operating in context $c$, and therefore suppress the associated indexing except when necessary for clarification.  
Thus, we define 
\begin{small}
	\begin{align*}
	&\mathcal{U}(\mathbb{P}(\pi)) = \\
	&\E \left\{ \pi(W_i) \E \left[ Y_{1i} - Y_{0i} - \frac{C(W_i,Y_{1i})}{\kappa} \bigg | W_i \right] + (1-\pi(W_i)) \E \left[ \frac{C(W_i,Y_{1i})}{\kappa} - (Y_{1i} - Y_{0i}) \bigg  | W_i \right]\right\} \\
	&= 2 \E \left\{ \pi(W_i) \E \left[  Y_{1i} - Y_{0i} - \frac{C(W_i,Y_{1i})}{\kappa}  \bigg | W_i  \right] \right\} -  \E \left[  Y_{1i}   - Y_{0i} - \frac{C(W_i,Y_{1i})}{\kappa}  \right],
	\end{align*}
\end{small}
where  $\kappa > 0$ is the level of cost effectiveness at which the planner is indifferent between allocating funds to treat an individual and using the funds for some alternative purpose \citep{Garber1997, Dhaliwal2011}, and $C(w,y) \ge 0$ is the  marginal cost function for treating an individual with covariate value $w$ and treated outcome $y$.  
Choice of a treatment assignment with this specification will be equivalent to operating under a budget constraint that must be satisfied in expectation {\it ex ante} (cf. \citealp{Bhattacharya2012}, versus {\it ex post}, as in \citealp{Kitagawa2015}), where the marginal value of relaxing the budget constraint is $1/\kappa$.  

The prediction method $m$ translates the planner's objective into an {\it ex ante} recommendation for a treatment assignment rule, $\pi_m$. Method $m$ does not have access to $\mathbb{P}$ and rather must rely on some approximation, $\widehat{\mathbb{P}}_m$.  As such, method $m$ solves,
\begin{align*}
\pi_m = \text{arg}\max_{\pi} \mathcal{U}(\widehat{\mathbb{P}}_{m}(\pi)),
\end{align*}
where $\widehat{\mathbb{P}}_{m}(\pi)$ can be informed by samples from the reference contexts and status quo data from the target.   Under our specification of the welfare function, we have
\begin{align*}
\pi_m = \text{arg}\max_{\pi } \widehat{\E}_m \left\{ \pi(W_i) \widehat{\E}_m \left[  Y_{1i} - Y_{0i} - \frac{C(w,Y_{1i})}{\kappa}   \bigg | W_i  \right] \right\}.
\end{align*}
In practice, we will use a ``plug-in'' rule where $\hat{E}_m[\cdot]$ is generated from model $m$ estimated on the reference data and status quo data from the target \citep{Hirano2019}. The problem that method $m$ solves generates a decision rule analogous to \cite{Manski2004}'s conditional empirical success rule: 
\begin{align*}
\pi_m(w) = 1 \left\{  \widehat{\E}_m \left[  Y_{1i} - Y_{0i} - \frac{C(w,Y_{1i})}{\kappa}  \bigg | w  \right] \geq 0  \right\}.
\end{align*}
In our application to conditional cash transfer programs, the cost function takes a form that allows for further simplifications.  Those details are presented below.

\section{{\it Ex Post} Inference}

The welfare contrast for two methods, $l$ and $m$, is given by
\begin{align*}
	\Delta_{lm} &=   \E\left\{ (\pi_l(W_i) - \pi_m(W_i))  \E \left[  Y_{1i} - Y_{0i} - \frac{C(W_i,Y_{1i})}{\kappa}  \bigg  | W_i  \right]   \right\},
\end{align*}
where $2 \Delta_{lm}$ would correspond to the difference in utilities as per our specification of $\mathcal{U}(\mathbb{P}(\pi))$ above. 
In constructing this {\it ex post} contrast, we condition on any data used to generate the $\hat{\mathbb{P}}$s.  
The welfare constrast $\Delta_{lm}$ is non-zero for values of $w$ for which the methods disagree on treatment assignment (i.e., $w$ for which $\pi_l(w) \ne \pi_m(w)$). 
It equals the marginalized value of the minimum-benefit-adjusted conditional average treatment effects when $m$ says to treat and $l$ says not to (or the reverse).

{\it Ex post}, we suppose that we obtain a random sample of experimental units in the target context for which conditions C1-C3 hold.  
We assume that in this {\it ex post} experiment, treatment assignment probabilities are given by $p(W_i)$, and that these probabilities are known. 
The experimental data in our target context allow us to estimate this welfare contrast.  
Given a random sample of size $N$ in the target context, we consider an estimator for the linear welfare contrast based on inverse-probability of treatment weighting with normalized weights.  
This estimator is efficient among consistent estimators that avoid modeling of either the potential outcome surfaces or conditional treatment probabilities (\citealp{Hirano2003}; \citealp{Lunceford2004}; \citealp[35]{Imbens2009}).  
We define the estimator as
\begin{align*}
\hat \Delta_{lm} = &\frac{\sum_{i = 1}^{N} \frac{T_i}{p(W_i)} \left(\pi_l(W_i) - \pi_m(W_i)\right)\left(Y_i - \frac{C(W_i,Y_i)}{\kappa}\right)}{\sum_{i = 1}^{N} \frac{T_i}{p(W_i)} } \\
& - \frac{\sum_{i = 1}^{N} \frac{(1-T_i)}{1-p(W_i)} \left(\pi_l(W_i) - \pi_m(W_i)\right)Y_i}{\sum_{i = 1}^{N}\frac{(1-T_i)}{1-p(W_i)}}.
\end{align*}
Inference for this estimator is based on the random sampling of $(Y_i, T_i, W_i)$ values from $\mathbb{P}$ under conditions C0-C3:

\begin{prop}\label{prop:Delta}
Under conditions C0-C3, as $N \rightarrow \infty$
$$
\frac{\sqrt{N}\left(\hat \Delta_{lm} - \Delta_{lm}\right)}{\sqrt{\hat V_{\Delta_{lm}}}} \overset{d}{\rightarrow} N(0,1),
$$
where
\begin{align*}
\hat V_{\Delta^g_{lm}} = & \frac{\sum_{i=1}^N \frac{T_i}{p(W_i)^2}\left[(\pi_l(W_i) - \pi_m(W_i)) (Y_i - \frac{C(W_i,Y_i)}{\kappa}) - \hat \delta_1\right]^2}{\sum_{i = 1}^{N} \frac{T_i}{p(W_i)}} \\
& + \frac{\sum_{i=1}^N \frac{1-T_i}{[1-p(W_i)]^2}\left[(\pi_l(W_i) - \pi_m(W_i))Y_i - \hat \delta_0\right]^2}{\sum_{i = 1}^{N} \frac{1-T_i}{1-p(W_i)}},
\end{align*}
and
\begin{align*}
&\hat \delta_ 1 = \frac{\sum_{i = 1}^{N} \frac{T_i}{p(W_i)} \left(\pi_l(W_i) - \pi_m(W_i)\right)(Y_i - \frac{C(W_i,Y_i)}{\kappa})}{\sum_{i = 1}^{N} \frac{T_i}{p(W_i)} } \\
&\hat \delta _0 = \frac{\sum_{i = 1}^{N} \frac{1-T_i}{1-p(W_i)} \left(\pi_l(W_i) - \pi_m(W_i)\right)Y_i}{\sum_{i = 1}^{N}\frac{1-T_i}{1-p(W_i)}}.
\end{align*}
\end{prop}
All proofs are contained in the appendix.  
Conditional on $W_i = w$, the recommendations, $\pi_m(w)$ and $\pi_l(w)$, are fixed.  
Our uncertainty about the welfare contrast is due to sampling and treatment assignment variation in the experimental data gathered in the target context that we use for the {\it ex post} assessment.\footnote{Treating the treatment assignment rules as fixed means that we can have expert opinion among the methods considered in $\mathcal{M}$. 
\citet{Diebold2015} makes this point in reviewing \citet{Diebold1995} and the literature following it, drawing a distinction between  \citet{Diebold1995} and \citet{West1996}, which additionally considers uncertainty arising from the samples on which models are fitted to generate predictions.}

The {\it ex post} situation that we consider is simple in that all randomization and sampling occurs at the unit level and there is no causal interference in the outcome data generating process.
Certainly the analysis could be extended to consider cluster randomization or cluster sampling, covariate adjustment, and targeting quantities that are identified under interference.  
For the present analysis, we stick with this relatively simple case.

Proposition \ref{prop:Delta} is sufficient to perform inference for any pair of methods.
\citet{Hansen2011} provide a sequential multiple testing algorithm for establishing a ``model confidence set'' (MCS) of level $1-\alpha$, which allows one to distinguish a set of best performing algorithms with an asymptotic error rate of $\alpha$.

\section{Empirical Application}

For an empirical illustration, we use data from Mexico and Morocco on the effects of conditional cash transfers (CCTs) on primary school enrollment.
We consider a policy scenario where a planner in Morocco is seeking recommendations for implementing a conditional cash transfer (CCT) program.
The planner's objective is defined as above, accounting for opportunity costs.
We use data from randomized evaluations of the PROGRESA  program in Mexico \citep{Schultz2004, Behrman2005, DeJanvry2006, Todd2006, Attanasio2012} and the TAYSSIR program in Morocco \citep{Benhassine2015}.
To construct the {\it ex ante} evaluations, we are limited to using the full data from Mexico and then the covariate and control group outcome data from Morocco.
The {\it ex post} assessment is done using the full experimental data from Morocco. 

\subsection{General Setup}
The outcome, $Y_{ic}$, is the school enrollment of child $i$. 
The covariate set, $\mathcal{W}_c$, is defined as the intersection pre-treatment characteristics on the PROGRESA and Moroccan questionnaires.  
The cost-effectiveness benchmark, $\kappa$, is based on gains from school enrollment. \cite{Montenegro2014} report a 10\% Mincer\footnote{controlling only for potential experience and its square.} earnings premium for each additional year of schooling in Morocco.  The average market earner in the \cite{Benhassine2015} sample gets \$1578.20 per year.  At a 5\% discount rate, the net present value of an additional year spent in school is approximately \$1,000.\footnote{Over 40 years of work.}
The treatment cost function, $C(W_i, Y_{1i})$, is based on the value of the conditional transfer. In Morocco, for 6-7 year olds, this amounts to the following:
$$
60 \text{ MAD per month} \times 12 \times  \frac{0.1\text{ years of ed.}}{100 USD} \times \frac{\text{1 USD}}{\text{8 MAD}} \times Y_{1i} = 0.09 Y_{1i} = \frac{C(w, Y_{1i})}{\kappa}.
$$
Correspondingly, for 8-9 year olds the transfer value is 80 MAD per month,  which means $0.12 Y_{1i} = \frac{C(w, Y_{1i})}{\kappa}$, and for 10+ year olds, the transfer is 100 MAD per month, and so $0.15  Y_{1i} = \frac{C(w, Y_{1i})}{\kappa}$.  If we define these age specific multipliers as $g(Age_i)$, then the objective for model $m$ reduces to,
\begin{align*}
	\pi_m &= \max_{\pi} \hat{\E}_m \left\{ \pi(W_i) \hat{\E}_m \left[ (1 - g(Age_i))Y_{1i} - Y_{0i} | W_i \right] \right\},
\end{align*}
generating the decision rule
\begin{align*}
	\pi_m(w) = 1 \left\{ \hat{\E}_m \left[ (1 - g(Age_i))Y_{1i} - Y_{0i} \bigg | w \right]  \geq 0 \right\}.
\end{align*}
This implies evaluating the signs of estimated conditional treatment effects on the adjusted outcome,
$$
Y^{adj}_{ic} = T_{ic}(1-g(\text{Age}_{ic})Y_{ic} + (1-T_{ic})Y_{ic}.
$$

We allow each method all observations and variables from PROGRESA ($D_c = 0$), and observations from TAYSSIR with $D_c = 1, T_{ic} = 0, U_{ic} \leq 0.5 $ where we assign $U_{ic} \sim U[0,1]$ independently from all other variables once (i.e., a random 50\% split of the TAYSSIR control group).  The methods use these data to compute $\pi_m(w)$ as defined above and then for methods $l$ and $m$, we compute $\widehat{\Delta}_{lm}$ as
\begin{align*}
	\hat \Delta_{lm} = &\frac{\sum_{\{i : D_c = 1, T_{ic} = 1\}} \frac{T_i}{p_c(W_i)} \left(\pi_l(W_i) - \pi_m(W_i)\right)(1 - g(Age_i))Y_i }{\sum_{\{i : D_c = 1, T_{ic} = 1\}} \frac{T_i}{p_c(W_i)} } \\
	& - \frac{\sum_{ \{i : D_c = 1, T_{ic} = 0, U_{ic} > 0.5  \} } \frac{(1-T_i)}{1-p_c(W_i)} \left(\pi_l(W_i) - \pi_m(W_i)\right)Y_i}{\sum_{\{i : D_c = 1, T_{ic} = 0, U_{ic} > 0.5\} }\frac{(1-T_i)}{1-p_c(W_i)}},
\end{align*}
where for the TAYSSIR experiment, $p_c(w)$ is known = $p_c$.  The estimate for the variance is constructed analogously.

\paragraph{Methods}

We consider methods that are already available from the current literature and that are straightforward to implement.    
We highlight assumptions on the joint distribution of random variables for each method.
However, we are not concerned with testing these assumptions directly, as in \cite{Allcott2012}, \cite{Dehejia2014}.
Instead, we list them as part of the specification of each method.
We are interested in assessing the relative empirical performance of the methods, all of which we view as likely misspecified \citep{Wolpin2007}.

The first two methods rely on reduced-form extrapolation of conditional treatment effects, as per, e.g., \citet{JosephHotz2005} and \citet{Dehejia2016}.  
The reduced form approaches we use include a ``low-tech'' version that simply takes the age-sex-specific treatment effects from PROGRESA and extrapolates them using the age-sex distribution in the TAYSSIR sample.    
We also use a ``high-tech'' approach that applies the generalized random forest estimator for heterogeneous treatment effects, as proposed by \citet{Athey2019-grf}.  

The next two methods include structural counterfactual predictions.  
This includes a static semi-parametric structural model for PROGRESA based on  \citet{Todd2010a} as well as a dynamic, parametric structural model for PROGRESA based on \citet{Attanasio2012}.

\subsection{Method 1: Extrapolating Age-Sex Conditional Effects}
The ``low-tech'' reduced form approach uses the adjusted outcomes ($Y^{adj}_{ic}$) from the PROGRESA data to estimate conditional treatment effects in age-sex strata for girls and boys with ages ranging from 6 to 16. 
These conditional treatment effects are then extrapolated to Morocco, which is justified given the following assumption.
\begin{description}
	\item[Age-Sex 1.] Unconfounded location (\citet{JosephHotz2005}) given age and sex.
	$$
		\E[(1-g(Age_i))Y_{1i} - Y_{0i} | D_c, Age_i, Sex_i] = \E[(1-g(Age_i))Y_{1i} - Y_{0i} | Age_i, Sex_i].
	$$
\end{description}

 Table \ref{tab:age-sex-cates-mex} shows the raw treatment effects for the age-sex strata in PROGRESA.
 Almost all of these are positive, so any treatment assignment based on them will be very liberal.
 We see exactly this in the second column of Table \ref{tab:agesex-all-welfare} which compares the effects of treatment assignments based on age-sex extrapolation to simply treating everyone, as a researcher might be inclined to recommend after seeing the large effect of TAYSSIR on enrollment.
 Without adjusting for cost-effectiveness, age-sex extrapolation recommends treating 97\% of the Moroccan sample.
  
 Table \ref{tab:age-sex-adj-cates-mex} lists the enrollment gains by age-sex strata, after adjusting for cost-effectiveness.
 The cost-effectiveness threshold we've imposed in Morocco is quite stringent: only a few subgroup effects remain positive.
 The first column of Table \ref{tab:agesex-all-welfare} shows the implications.
 Now the age-sex extrapolation only recommends treating 13\% of children in Morocco.
 This results in a statistically significant 3.5 percentage point decrease in enrollment gain (the third row and forth row), but a 6.5 percentage point increase in cost-effectiveness-adjusted enrollment gain (the fifth and sixth row).
The increase comes from the fact that age-sex based-extrapolation avoids recommending treatment for younger children whose enrollment is already almost universal.
Universal untreated enrollment has to be paid for, which is expensive, and leaves little room for enrollment gain so the cost is not worth it.

\subsection{Method 2: Generalized Random Forest-Based Extrapolation}
\label{sec:grf}
The ``high-tech'' reduced form approach fits generalized random forest (GRF) models to the adjusted outcomes in the PROGRESA data.
At present we include as covariates the child's pre-treatment enrollment status, years of education, literacy status, age, and sex, and then for the child's household, the number of children, whether the head is male, whether it is single-parent, whether the father is alive, whether the mother is alive, whether the child lives with the father, or whether the child lives with the mother.
In cases of item-level missing data, we impute a zero for the missing value and then accompany the variable with a separate indicator variable for whether the value is missing.  
Defining $W_i$ the vector of variables described above, GRF-based extrapolation is justified under the following assumption.
\begin{description}
	\item[GRF1.] Unconfounded location given $W_i$.
	$$
	\E[(1-g(Age_i))Y_{1i} - Y_{0i} | D_c, W_i ] = \E[(1-g(Age_i))Y_{1i} - Y_{0i} | W_i ].
	$$
\end{description}

Using these variables, we fit a generalized random forest (GRF) using the algorithm written by \citet{Athey2019-grf}.
We set the number of trees to 2000 with a minimal leaf size of 2 units.  
We use 50\% splits of the data to build trees on one random split of the data and then select the error-minimizing tree pruning by evaluating predictions with different candidate pruning levels on the data from the other split. 
These settings do not depart very much from the defaults set by \citet{Athey2019-grf}.

There are two approaches to using GRFs to characterize effect heterogeneity.  
The first is to fit GRFs to the treated and control outcomes separately, and then combine those fits to construct estimates of unit-level treatment effects.
This approach targets loss on the level of treated and control outcomes, and then indirectly targets effect heterogeneity.
The second is to fit a GRF with a loss function that is specifically targeted to effect heterogeneity.  
In simulations and various empirical tests on the PROGRESA sample, we found that the first approach was substantially better for our application. 
This was surprising, but the reason seemed to be that the approach targeting effect heterogeneity directly tended to regularize too heavily, and therefore did not discriminate strongly enough between classes of units with large or small (or even negative) effects.
As such, our analysis applies the method that models the treated and control outcomes separately.

Tables \ref{tab:mexico-grfpo-noce-varimp} and \ref{tab:mexico-grfpo-ce-varimp} show variable importance summaries provided by \citet{Athey2019-grf} for random forests fitted to enrollment without and with cost-effectiveness adjustment, respectively.
The summaries essentially count instances where a variable is used in a split, where instances at the top end of the tree get higher weight than at the lower end. 
Age, baseline enrollment, and number of years of education completed emerge as most important but several other variables matter as well, including parental education and presence.

Tables \ref{tab:morocco-grfpo-noce-catesum-boys} and \ref{tab:morocco-grfpo-noce-catesum-girls} display treatment effects on enrollment, aggregated by age-sex strata to allow for comparisons to the ``low-tech'' reduced form approach.
The GRF detects a great deal of treatment effect heterogeneity within strata, with positive and negative effects for all strata.
Table \ref{tab:morocco-grfpo-noce-treat-rate} show the share in each stratum recommended for treatment, which increases quickly with age.

Tables \ref{tab:morocco-grfpo-ce-catesum-boys} and \ref{tab:morocco-grfpo-ce-catesum-girls} gives treatment effect predictions within strata, after adjusting the treated outcome for cost-effectiveness.
The average treatment effect is predicted to be negative in most strata, with the exception of those including children aged 13 and up.
Table \ref{tab:morocco-grfpo-ce-treat-rate} shows the treatment recommendations.
Despite negative average effects within strata, the treatment rates are still quite high.

Table \ref{tab:grf-agesex-welfare} evaluates GRF's treatment recommendations relative to simple age-sex-based extrapolation from PROGRESA.
The results are perhaps surprising.
GRF treats more individuals, resulting in 1 percentage point more enrollment gain.
But the cost outstrips the gain and the adjusted enrollment gain is negative and statistically significant.

We believe this disappointing performance of GRF versus a much simpler method is due to a kind of ``contextual overfitting''.
While the GRF's regularization guards against within-context overfitting, the dimensions of heterogeneity generating a large share of positive treatment effects even after adjusting for cost-effectiveness may represent quirks of PROGRESA's implementation in Mexico.
It remains to be seen whether this problem persists when we add reference contexts and context characteristics to our exercise.

\subsection{Method 3: Semi-Parametric Structural Approach}
\label{sec:sps_model}
\subsubsection{\citet{Todd2010a}'s Non-Parametric Structural Model}
The third method takes as a starting point the non-parametric structural (NPS) model of school attendance proposed in \citet{Todd2010a}.
For parsimony in the non-parametric step of our own approach (described in the next section) we use the simplest version of the NPS model, which considers the enrollment decision for each child independently.
In this version of the model, households solve the following static utility maximization problem with utility depending on household consumption $c$, the child's enrollment status $y$.
\begin{align}
	&\max_{y \in \{0, 1\}} \ U(c,y;w,\epsilon) \notag \\
	\label{eq:nps_budget_constraint}
	&\text{subject to } c = n + e(1-y).
\end{align}
Abusing notation slightly, $W_i$ is here redefined as $W_i \setminus E_i, N_i$ where $E_i$ represents the child's wage offer, $N_i$ household income for child $i$ excluding $i$'s own earnings.
The redefined $W_i$ and $V_c$ are observed shifters in preference for child schooling and $\epsilon_i$ is an unobserved shifter.

Optimal school attendance is given by
$$
	y_0 = \phi(n,e; w, \epsilon) = 1\{U(n,1;w,\epsilon) > U(n+e,0; w, \epsilon)\}.
$$
Now modify the budget constraint by introducing the treatment, a grant $g$ paid only when the child attends school:
\begin{align*}
	c&= n+ e(1-y) + g y \\
	&= (n + g) + (e - g)(1-y).
\end{align*}
It is easy to see that $y_1$, optimal school attendance with the grant program in place, can be obtained by plugging a modified version of non-child income ($n+g$) and modified child wage offer ($e-g$) into the same $\phi(\cdot)$ function as before: 
$$
	y_1 = \phi(n+g, e-g; w, \epsilon).
$$

In addition to this structure, NPS imposes a key assumption on the distribution of random variables:
\begin{description}
	\item[NPS1.] Exogeneity of non-child income and child wage offers: $\epsilon_{i} \independent N_i, E_i  | W_i$.
\end{description}
The NPS CATE for children with characteristics $w, n, e$ is identified as
\begin{align}
\notag
\E [ Y_{1i} | N_i = n, E_i = e, W_i = w] - & \E [ Y_{0i} | N_i=n, E_i=e, W_i = w ] \\
\notag
= \E[ Y_{0i} | N_i = n+g, E_i = e - g, W_i = w] - & \E[ Y_{0i} | N_i=n, E_i=e, W_i = w] \\
\label{eq:nps_cate}
= \E[ Y_i | T_i = 0, N_i = n+g, E_i = e - g, W_i = w] - & \E[ Y_i | T_i=0, N_i=n, E_i=e, W_i = w ] 
\end{align}
Crucially, \eqref{eq:nps_cate} is identified in the data provided to predictors for context $c$.

\subsubsection{Our Implementation: the Semi-Parametric Structural (SPS) Approach}
As in \citet{Todd2010a}, we do not observe wages for most children in our contexts.
In addition, in some of our contexts non-child income is not directly observable since households are smallholder farmers.
We therefore make some modifications to the NPS model.

Each child's wage offer is given by exponentiating the conditional expectation of her earnings given her age,  gender, industry of work\footnote{If observed, otherwise we set this to the most common industry for child workers.}, and locality of residence.\footnote{\citet{Todd2010a}, in contrast set children's wage offers to the village-level wage for an agricultural worker but this variable is not available in many of our contexts, and in fact in only half of the PROGRESA evaluation villages they consider.}
\begin{description}
	\item[SPS1.] Conditional expectation wage offers:
	\begin{align*}
		E_i &= \exp( \E [ \log(E_i) | Age_i, Sex_i, Industry_i, Locality_i ] ) \\
		  	&= e_0(Age_i, Sex_i, Industry_i, Locality_i).
	\end{align*}
\end{description}
This is very close to how we handle wages in the dynamic parametric structural model (DPS) described in \ref{sec:dps_model}, except that in the DPS model agents optimize with respect to (1) the sample mean function $\hat{\E}[ E_i | Age_i, Sex_i, Industry_i, Locality_i ]$ rather than (2) the population expectation function $ \E[ E_i | Age_i, Sex_i, Industry_i, Locality_i ]$ (use of (1) in DPS exactly follows AMS).
We handle non-child income in the same way: by computing the sum of the expected earnings of all the members in $i$'s household, excluding $i$.
\begin{description}
	\item[SPS2.] Conditional expectation non-child income:
	\begin{align*}
		N_i &= \sum_{j \in Household_i, j \neq i} \exp ( \E [ \log(E_j) | Age_j, Sex_j, Industry_j, Locality_i ] ) \\
			&= n_0(Household_i).
	\end{align*}
	Additionally, assume $E_i$ is missing at random, effectively following \cite{Todd2010a}.
\end{description}

With industry and locality indicators, the set of conditioning variables is high-dimensional.
We therefore estimate the conditional expectation function by LASSO, which is justified under the following assumption.
\begin{description}
	\item[SPS3.] Approximately sparse linear representation of expected wages.
	\begin{align*}
		&\E [ \log(E_i) | Age_i, Sex_i, Industry_i, Locality_i ] \\
		&\approx \eta_0 + \eta_1 Age_{i} + \eta_2 (Age_{i} - 21 )_{+} + \eta_3 male_i + \zeta_{industry} + \lambda_{province} + \xi_{locality} + \nu_{i}
	\end{align*}
	where approximation is in the sense of \cite{Belloni2012}.
	The notation $(\cdot)_{+}$ indicates the positive part of the expression in parentheses.
	$\zeta_{industry}, \lambda_{province}$, and $\xi_{locality}$ are fixed effects.
	Province is the top level subnational geographic unit (like a US state), which is defined for each context following IPUMS-International.
	Locality is the smallest geographic level available in each dataset.
	We do not subject the linear spline in age to the LASSO penalty because the substantial positive gradient of wage in age for youths is a key driver of the opportunity cost of enrollment.
	Similarly, we always include the industries employing the majority of children in each context.\footnote{We exclude education following AMS. In practice, and echoing AMS's findings, education has little effect on earnings in our rural contexts. Education is never selected by LASSO and its inclusion has almost no impact on other variables coefficients or the selected penalty term.} \label{ass:sparse_regression}
\end{description}
We select the LASSO penalty term by 5-fold least squares cross-validation.
Finally, we replace assumption NPS1 with 
\begin{description}
	\item[SPS4.] $\E[\varepsilon_{i} | e_0(Age_i, Sex_i, Industry_i, Locality_i), n_0(Household_i), Sex_i] = 0$ where $\varepsilon_i$ is defined in the conditional expectation equation
	\begin{align}
		\label{eq:sps_exogeneity}
		Y_i = m_0(e_0(Age_i, Sex_i, Industry_i, Locality_i), n_0(Household_i), Sex_i) + \varepsilon_{i}.
	\end{align}
\end{description}

Our approach is the following.
We estimate $e_0(Age_i, Sex_i, Industry_i, Locality_i)$ in a first stage and use the resulting $\hat{e}$ in place of $e_0$ in $n_0$.
We then estimate $m_0$, our analog of $\E[ Y_i | T_i=0, N_i=n, E_i=e, W_i = w ] $ from equation \eqref{eq:nps_cate}, with $\hat{e}$ and $\hat{n}$ in place of $e_0$ and $n_0$.
$\hat{e}$ and $\hat{n}$ are generated regressors in the sense of \cite{Mammen2012} who show that the second stage estimate is consistent for $m_0$.
Our second stage implementation is by mixed datatype kernel regression with bandwidth selected by cross-validation using the \texttt{np} package in R (\cite{Hayfield2008}).
We then use the second stage estimate to generate the predicted SPS CATE just as in equation \eqref{eq:nps_cate}.

\subsubsection{Results}
Figures \ref{fig:sps_original_density}, \ref{fig:sps_counterfactual_density}, and \ref{fig:sps_bivariate_regression} provide some intuition for how SPS works in practice.
Figure \ref{fig:sps_original_density} shows the bivariate density of child wage offers in the held out portion of the control group on the y-axis and non-child income on the x-axis.
Figure \ref{fig:sps_counterfactual_density} shows how the model predicts the effect of the TAYSSIR grant.
Effective child wage offers decrease, shifting the density down, and non-child income is increased slightly, moving density slightly right.

The impact of this shift depends on the non-parametric regression of enrollment as a function of child wage offer and non-child income in the portion of the control group (non-holdout) available to predictors, which is depicted in Figure \ref{fig:sps_bivariate_regression}.
The plot shows that non-child income has little association with enrollment conditional on a child's wage offer.
child wage offers do matter, however, particularly in the region where probability mass is being moved.

Table \ref{tab:sps-agesex-welfare} shows that the SPS approach tends to overstate the magnitude of treatment effects.\footnote{This turns out to be true for both positive and negative treatment effects, with specifics to be included in later drafts.}
Relative to age-sex based extrapolation from PROGRESA, SPS is only slightly better in cost-adjusted enrollment gain terms than making all children in Morocco eligible for TAYSSIR grants.
We think this is because the non-parametric second step is simply too volatile, despite having been regularized by cross validation based bandwidth selection.

\subsection{Method 4. Dynamic Parametric Structural}
\label{sec:dps_model}
\paragraph{Model}
Our Dynamic Parametric Structural (DPS) model is largely based on \citet{Attanasio2012}, with a few modifications made to fit data availability and improve in-sample fit in our contexts, following standard practice (see e.g. \citet{Wolpin2013}).
The main dynamic features of the model are (1) a finite horizon - children can only be enrolled until age 17 and so can only accumulate subsidies up to this point and (2) persistence of education choices - the flow utility of enrollment is affected by the number of years the child is behind her age-appropriate grade level.
Uniquely among our methods, the DPS model allows modeling of the entire subsidy schedule by age.

\paragraph{Flow utility}
The in-period utility for child $i$\footnote{Again, schooling decisions are made by the household for each child independently.} at age $a$ in school and work are $u^S_{ia}$ and $u^W_{ia}$, respectively:
\[
	\begin{array} { l } { u _ { ia } ^ { S } = \gamma \delta g _ { ia } + \mu_i + \psi ^ { \prime } z _ { ia } + 
	b' \cdot \text{yrs\_behind}_{ia} + 
	1 \left( p _ { ia } = 1 \right) \beta ^ { p } x _ { ia } ^ { p } + 1 \left( s _ { ia } = 1 \right) \beta ^ { s } x _ { i a } ^ { s } + \varepsilon _ { i a } } \\ 
{ u _ { ia } ^ { W } = \delta w _ { ia } } \end{array}\]
$g_{ia}$ represents the grant $i$ is entitled to given her completed years of schooling and other characteristics (for example gender).
$\mu_i$ is a child-specific shifter to the preference for enrollment, drawn from a discrete distribution with $K$ points of support.
We will refer to each point of support as an unobserved child ``type''.
In practice, we estimate the model with three types.
$z_{ia}$ includes other observed covariates affecting preference for enrollment.
Specifically it includes $a$ and a dummy variable for whether $i$'s father received any formal schooling.
$yrs\_behind_{ia}$ is a three-element vector of dummy variables with indicators for being behind grade level by 1 year, 2 years, or $\geq 3$ years.
$p_{ia}$ is a dummy variable measuring whether $i$'s years of schooling at age $a$ make her eligible for primary school and $x^p_{ia}$ measures distance to school, proxying for cost of attendance.
$s_{ia}$ is a dummy variable equal to one if $i$'s years of schooling make her eligible for secondary school and $x^s_{ia}$ is a constant.
$e_{ia}$ is $i$'s wage offer.
$\varepsilon_{ia}$ is an IID idiosyncratic shock to the utility of $i$'s attending school at age $a$, which follows the logistic distribution.\footnote{Note that this $\varepsilon_i$ has no relation to the $\varepsilon_i$ from equation \eqref{eq:sps_exogeneity}.}

We allow some of the coefficients to vary by $i$'s sex.
In particular, $\beta^s$ and the component of $\psi$ multiplying age vary by gender.
The unobserved type probability also depends on the sex of the child.

\paragraph{Wage offer}
Log wage offers are computed according to the same sparse linear regression representation equation described in assumption SPS3 of the SPS model.
The only difference is that in the SPS model predicted values from the LASSO regression are treated as estimates of the true wage offer.
In the DPS model, agents are assumed to use the same predicted values as the econometrician.

\paragraph{Terminal value}
The value of having accumulated $ed_{i18}$ years of school in the terminal period (when the child is 18) is given below.
\[
	V \left( ed_{ i , 18 } \right) = \frac { \alpha _ { 1 } } { 1 + \exp \left( - \alpha _ { 2 }  ed_ { i , 18 } \right) }
	+ \alpha_3 \cdot 1\{ed_{i,18} \geq sec \}.
\]
$\alpha_3$ measures the additional value of having completed secondary school, measured by $ed_{i,18}$ being greater than the last year of secondary school, $sec$.

\paragraph{Value functions}
The value of choosing to have $i$ attend school after having completed $ed_{ia}$ years of education by age $a$ is:
\begin{align}
	V _ { i a } ^ { S } \left( e d _ { i a } \right) = u _ { i a } ^ { S } + \beta \{ p_a^S & \left( e d _ { i a } + 1 \right) \E \max \left[ V _ { i, a + 1 } ^ S \left( e d _ { ia } + 1 \right) , V _ { i, a + 1 } ^ { \mathrm { W } } \left( e d _ { i a } + 1 \right) \right] \notag \\
	& + \left( 1 - p_a ^S  \left( e d _ { i t } + 1 \right) \right) \E \max \left[ V _ { i, a + 1 } ^ { S } \left( e d _ { i a } \right) , V _ { i, a + 1 } ^W \left( e d _ { i t } \right) \right] \}. \label{eq:enrolling_value}
\end{align}
We implicitly condition on covariates $W_{ia}$ redefined as $W_{ia} \setminus ed_{ia}$ and $\mu_i$.
$p^s_a( ed )$ is the probability of successfully passing grade $ed$ at age $a$ conditional on enrolling.
We estimate it non-parametrically, outside the model (like the wage offer function).
If $i$ successfully passes the grade, she expects to receive the maximum of the value of enrolling or choosing to work in the next year with her education equal to $ed_{ia} + 1$.
The expectation is taken over possible realizations of the $\varepsilon_{ i, a + 1}$ shock.
If $i$ does not pass, she expects to get the maximum of the value of enrolling/working being one year older and with education still equal to $ed_{ia}$.
The term in braces in equation \ref{eq:enrolling_value} is thus the next-period expected value.
From the point of view of this period, the expected value is discounted by $\beta$ which we set equal to 0.95, following AMS.
We add the flow utility of being enrolled to complete the definition of the value function when enrolling.
The value of having $i$ work this period is simpler since $ed_{ia}$ stays fixed:
\[
	V _ { i t } ^W \left( ed_{ia} \right) = u _ { i t }^W + \beta \E \max \left\{ V _ { i, a + 1 }^S \left( ed_{ ia } \right) , V_{ i, a + 1 }^W \left(  ed_{ i a} \right) \right\}.
\]
Given a set of parameters, we solve for the outputs each value function (enroll, work) for of all possible combinations of age and years of education completed by backward induction, beginning by calculating the terminal value for each set of candidate parameters.\footnote{Note that since the only error term $\varepsilon_{it}$ follows an IID logistic distribution, then $\E \max \left\{ V _ { i t + 1 } ^ { s } \left( e d _ { i t } \right) , V _ { i t + 1 } ^ { \mathrm { w } } \left( e d _ { i t } \right) \right\}$ has a closed form (see \citet{Keane2011}).
	Our closed form is slightly different from theirs because they use two Type 1 Extreme Value random variables instead of one logistic draw.
	We simply subtract $\varepsilon_2$ in both equations in their function to derive our closed form:
	\[ 
		\E \max \left\{ V _ { i a }^S \left( e d _ { i a } \right) , V _ { i a } ^W \left( e d _ { i a } \right) \right\} = 
		\rho \{ \log[(\exp( V^S_{ia}) / \rho + \exp(V^W_{it}) / \rho] \}.
	\]
We normalize the scale $\rho$ of the error term to 1 in estimation.}

\paragraph{Likelihood}
With the value function in hand and the logistic error distribution assumption, it is straightforward to compute the likelihood of each child's being enrolled given observed characteristics $W_i$ and unobserved type $\mu_i$: $\Pr(Y_i = 1 | W_i, \mu_i = \mu_k)$.\footnote{We drop the $a$ subscript because the model is estimated on a single cross-section so age does not vary with $i$.}
The full likelihood associated with $Y_i = 1$ is given by:\footnote{If $Y_i = 0$ the likelihood contribution is 1 - \eqref{eq:enrolled_likelihood}.}
\begin{align}
	\label{eq:enrolled_likelihood}
	\Pr(Y_i = 1| W_i ) = \sum_{k=1}^K
	\frac{1}{1+\exp(V^W_{it}(ed_i | \mu_k, W_i )-V^S_{it}( ed_i| \mu_k, W_i))}
	\Pr(\mu_i = \mu_k | W_i, \mu_k).
\end{align}

Since the utility shock $\varepsilon_i$ is IID across time and individuals, we could in principle derive the conditional type probability as a function of the history of characteristics and education decisions of $i$ starting at the age when $i$'s enrollment was first considered ($min\_a$): $(W_{i,a-1}, \ldots, W_{i,min\_a})$ and $(Y_{i,a-1}, \ldots, Y_{i,min\_a})$ respectively:
\begin{align}
	&\Pr(\mu_i = \mu_k | (Y_{i,a-1}, \ldots, Y_{i,min\_a}), (W_{i,a-1}, \ldots, W_{i,min\_a})) \notag \\
	&= \frac{\Pr(Y_{ia} | W_{ia}, \mu_i = \mu_k) \cdots \Pr(Y_{i,min\_a} | W_{i,\min\_a}, \mu_i = \mu_k) \Pr (\mu_k)}
{\sum_{k=1}^{K} \Pr(Y_{ia} | W_{ia}, \mu_i = \mu_k) \cdots \Pr(Y_{i,min\_a} | W_{i,\min\_a}, \mu_i = \mu_k) \Pr (\mu_k)}. \label{eq:type_probability}
\end{align}
Actually estimating \eqref{eq:type_probability} is infeasible, however, since it requires knowledge of the full histories $(W_{i,a-1}, \ldots, W_{i,min\_a})$ and $(Y_{i,a-1}, \ldots, Y_{i,min\_a})$\footnote{These would be included in $W_i$, but unfortunately we do not have them.} and, furthermore, would be very high-dimensional.

Following \cite{Todd2006}, we instead use a multinomial logit approximation:
\begin{align*}
&\Pr(\mu_i = \mu_k | (Y_{i,a-1}, \ldots, Y_{i,min\_a}), (W_{i,a-1}, \ldots, W_{i,min\_a})) \\
&\approx \Pr(\mu_i = \mu_k | a, ed_{ia}, gender, father\_ed) \\
& \ \ \ \approx \frac{\exp(\beta^\mu_k \cdot (1, a, ed_{ia}, gender, father\_ed))}
{1 + \sum_{k} \exp(\beta^\mu_k \cdot (1, a, ed_{ia}, gender, father\_ed)}, k \in {1 \ldots, K-1}
\end{align*}
Age proxies for the length of the history, $ed_{ia}$ for $(Y_{i,a-1}, \ldots, Y_{i,min\_a})$ and $gender$ and $father\_ed$ for $(W_{i,a-1}, \ldots, W_{i,min\_a})$ (specifically $W_{i,min\_a}$).
Importantly, $ed_{ia}$ is excluded from flow utility so there is independent variation to identify the conditional type probabilities.
We estimate the parameters described above, along with the support points $(\mu_1, \cdots, \mu_K)$ by maximum likelihood.
 
\subsubsection{Results}
Figures \ref{fig:dps_fit_by_age} and \ref{fig:dps_fit_by_ed} show that the model fits well in the portion of the Moroccan control group made available to predictors.
Figure \ref{fig:dps_fit_by_age} shows the fit to enrollment rates by age.
The size of each point on the graph representing the sample size in that age-sex stratum.
The model captures delayed entry into school, near-universal enrollment at young ages, and the sharp drop in enrollment for teenagers.
Figure \ref{fig:dps_fit_by_ed} illustrates the fit by number of years of education completed.
It captures the drop in enrollment at the transition to secondary school (year 7 in Morocco).

Table \ref{tab:dps-agesex-welfare} shows the results from using the DPS model to predict the effect of the TAYSSIR program.
For all but a tiny fraction of the Moroccan holdout sample, predicted enrollment gains - while reasonable - are too small to exceed the cost-effectiveness threshold.
We show this visually in Figure \ref{fig:dps_predicted_te_dist}, plotting the CDF of predicted treatment effects due to observables.
A key point from Table \ref{tab:dps-agesex-welfare} is how close TAYSSIR is to being non-cost-effective for \emph{any} child.
Age-sex based extrapolation from PROGRESA only provides a small, statistically insignificant increase in welfare relative to DPS's no-treatment recommendation.

\section{Conclusion}

We develop a decision-based approach to comparing the relative performance of methods for generating counterfactual predictions that are then used to make policy recommendations.
We consider a social planner who is operating in a target context and is seeking recommendations on what policy to choose from a set of feasible options.
The richness of the space of policy options determines the nature of the recommendations being sought  --- e.g., whether a simple up-or-down recommendation to treat everyone or no one, or a more refined recommendation about who should be treated and who not.
Recommendations could be based on econometric estimates, whether reduced form or structural, or expert opinions.  
Our leading application is one where the planner maximizes a linear welfare objective in assigning treatments on the basis of available covariate information.
In this case, the success of a method for generating recommendations depends on how accurately it can predict conditional treatment effects in the target context.

We define a welfare contrast to use for conducting an {\it ex post} analysis of how well different methods performed with respect to the planner's goals.  
We estimate this welfare contrast by using experimental data that reveals how a treatment affects the outcome distribution in the target population.  
The welfare contrast is straightforward to compute, and it allows us to judge whether one method outperforms another in a manner that is statistically significant.

We provide an empirical illustration that considers a planner seeking a recommendation on how to implement program using conditional cash transfers (CCTs) to boost school enrollment in Morocco.
The data available for generating recommendations include a randomized evaluation of CCTs in Mexico as well as data from Moroccan households under the status quo {\it ex ante}, in which no CCTs have been applied.
We generate recommendations from reduced form methods and structural models.
We then perform an {\it ex post} evaluation of these methods using data from a randomized evaluation of CCTs in Morocco. 
We view this toy example as helping build intuition for how to specify methods to evaluate in a full-featured empirical portion of the paper including the contexts from \cite{Banerjee2015b} which we will pre-specify.

We see this exercise as making three contributions.  
First, as our application attempts to show, it provides a clear framework to assess internal validity versus external validity trade-offs.
In particular, our application allows us to assess how robust and internally valid estimates from external contexts fare relative to within-context estimates that may be biased due to model misspecification  \citep{Pritchett2013}.
Second, it provides a principled basis for assessing the performance of different methods by tying the assessment to welfare considerations.  
This is important, because different objective functions can imply different rank orderings of methods. 
Our approach thus forces one to first consider the welfare objective so as to be clear about the relevant objective.
Third, we show that each experiment or observational study may contain much more decision-relevant information than would be contained in a single treatment effect estimate.  

We are undertaking a number of extensions to what we have done here.
Model selection or model averaging approaches based on our welfare criteria may lead to better predictions.
We also plan to work with evidence bases that include more external contexts.
In doing so, we would want to account for site selection, as per \citet{Allcott2015} and \citet{Gechter2018}.     
 
\clearpage
\appendix

\section{Proofs}

\begin{proof}[Proof of Proposition \ref{prop:Delta}] By the weak law of large numbers, Slutsky's theorem, and conditions C2 and C3, $\hat \Delta^g_{lm}$ has the same limit as
\begin{align*}
\tilde \Delta^g_{lm} = & \frac{1}{N}\sum_{i=1}^N \frac{T_i}{p(W_i)}\left(\pi_l(W_i) - \pi_m(W_i) \right)g(Y^P_i(1))\\ & - \frac{1}{N}\sum_{i=1}^N \frac{1-T_i}{1-p(W_i)}\left(\pi_l(W_i) - \pi_m(W_i) \right)g(Y^P_i(0)).
\end{align*}
Take the first term on the right-hand side. 
By the weak law of large numbers, iterated expectations, and condition C1,
\begin{align*}
\frac{1}{N}\sum_{i=1}^N \frac{T_i}{p(W_i)}\left(\pi_l(W_i) - \pi_m(W_i) \right)g(Y^P_i(1)) & 
\\ & \hspace{-5em} \overset{p}{\rightarrow} \E\left[\E[T_i|W]\frac{1}{p(W)}\E[\left(\pi_l(W_i) - \pi_m(W_i) \right)g(Y^P_i(1))|W] \right]\\
& \hspace{-5em} = \E\left[\left(\pi_l(W_i) - \pi_m(W_i) \right)g(Y^P_i(1)) \right],
\end{align*}
and similar for the second term.  
Thus as $N \rightarrow \infty$, $\E[\hat \Delta^g_{lm} - \Delta^g_{lm}] \overset{p}{\rightarrow} 0$.  
Having established that $\hat \Delta^g_{lm}$ is asymptotically unbiased for $\Delta^g_{lm}$, inference follows from the usual generalized method of moments results \citep{Newey1994, Lunceford2004}.
To see this, first note that $\hat \Delta^g_{lm} = \hat \delta_1 - \hat \delta_0$ for $(\hat\delta_1, \hat\delta_0)$ that solve the score equations
\begin{align*}
\sum_{i=1}^N \psi_1(\hat\delta_1) = 0 & \text{ and } \sum_{i=1}^N \psi_0(\hat\delta_0) = 0, \\
& \text{where } \psi_1(\hat\delta_1) = \frac{T_i\left[(\pi_l(W_i) - \pi_m(W_i))g(Y_i) - \hat\delta_1\right]}{p(W_i)}, \\
& \text{ and } \psi_0(\hat\delta_0) = \frac{(1-T_i)\left[(\pi_l(W_i) - \pi_m(W_i))g(Y_i) - \hat\delta_0\right]}{1-p(W_i)}.
\end{align*}
Then, given random sampling, bounded first and second moments, and conditions C1-C2, we have
\begin{equation}
\frac{\sqrt{N}\left(\hat \Delta^g_{lm} - \Delta^g_{lm}\right)}{\sqrt{V_{\Delta^g_{lm}}}} \overset{d}{\rightarrow} N(0,1), 
\end{equation}
where 
\begin{align*}
V_{\Delta^g_{lm}} & = \E[\psi_1(\delta_1)^2 + \psi_0(\delta_0)^2]\\
& = \E\left\{\frac{\left[(\pi_l(W_i) - \pi_m(W_i))g(Y_i^P(1)) - \delta_1\right]^2}{p(W_i)} + \frac{\left[(\pi_l(W_i) - \pi_m(W_i))g(Y_i^P(0)) - \delta_0\right]^2}{1-p(W_i)} \right\},
\end{align*}
with $\delta_t = \E[(\pi_l(W_i) - \pi_m(W_i))g(Y_i^P(t))]$.  Then, by the same conditions for which $\hat \Delta^g_{lm}$ is consistent for $\Delta^g_{lm}$, $\hat V_{\Delta^g_{lm}}$ is consistent for $V_{\Delta^g_{lm}}$.
\end{proof}

\clearpage

\singlespacing
\bibliographystyle{chicago}
\bibliography{../../post-feb-2018}

\begin{thebibliography}{}

\bibitem[\protect\citeauthoryear{Allcott}{Allcott}{2015a}]{Allcott2012}
Allcott, H. (2015a).
\newblock {Site Selection Bias in Program Evaluation}.
\newblock {\em Quarterly Journal of Economics\/}~{\em 130\/}(3), 1117--1165.

\bibitem[\protect\citeauthoryear{Allcott}{Allcott}{2015b}]{Allcott2015}
Allcott, H. (2015b).
\newblock {Site Selection Bias in Program Evaluation}.
\newblock {\em The Quarterly Journal of Economics\/}~{\em 130\/}(3),
  1117--1165.

\bibitem[\protect\citeauthoryear{Andrews and Oster}{Andrews and
  Oster}{2017}]{Andrews2017}
Andrews, I. and E.~Oster (2017).
\newblock {WEIGHTING FOR EXTERNAL VALIDITY}.

\bibitem[\protect\citeauthoryear{Aronow and Samii}{Aronow and
  Samii}{2017}]{aronow_samii2017_interference}
Aronow, P.~M. and C.~Samii (2017).
\newblock Estimating average causal effects under general interference, with
  application to a social network experiment.
\newblock {\em Annals of Applied Statistics\/}~{\em 11\/}(4), 1912--1947.

\bibitem[\protect\citeauthoryear{Arrow}{Arrow}{1972}]{Arrow1972}
Arrow, K.~J. (1972).
\newblock {The Value of and Demand for Information}.
\newblock In C.~McGuire and R.~Radner (Eds.), {\em Decision and Organization. A
  Volume in Honor of Jacob Marschak\/} (Amsterdam ed.)., Chapter~6, pp.\
  131--140. North-Holland.

\bibitem[\protect\citeauthoryear{Athey, Tibshirani, and Wager}{Athey
  et~al.}{2019}]{Athey2019-grf}
Athey, S., J.~Tibshirani, and S.~Wager (2019).
\newblock Generalized random forests.
\newblock {\em The Annals of Statistics\/}~{\em 47\/}(2), 1148--1178.

\bibitem[\protect\citeauthoryear{Athey and Wager}{Athey and
  Wager}{2017}]{Athey2017a}
Athey, S. and S.~Wager (2017).
\newblock {Efficient Policy Learning}.

\bibitem[\protect\citeauthoryear{Attanasio, Meghir, and Santiago}{Attanasio
  et~al.}{2012}]{Attanasio2012}
Attanasio, O., C.~Meghir, and A.~Santiago (2012).
\newblock {Education Choices in Mexico: Using a Structural Model and a
  Randomised Experiment to Evaluate PROGRESA}.
\newblock {\em Review of Economic Studies\/}~{\em 79\/}(1), 37--66.

\bibitem[\protect\citeauthoryear{Banerjee, Chassang, and Snowberg}{Banerjee
  et~al.}{2016}]{Banerjee2016b}
Banerjee, A., S.~Chassang, and E.~Snowberg (2016).
\newblock {Decision Theoretic Approaches to Experiment Design and External
  Validity}.
\newblock In {\em Handbook of Field Experiments, forthcoming}.

\bibitem[\protect\citeauthoryear{Banerjee, Hanna, Kreindler, and
  Olken}{Banerjee et~al.}{2017}]{Banerjee2015b}
Banerjee, A., R.~Hanna, G.~Kreindler, and B.~A. Olken (2017).
\newblock {Debunking the Stereotype of the Lazy Welfare Recipient: Evidence
  from Cash Transfer Programs Worldwide}.
\newblock {\em The World Bank Research Observer\/}~{\em 32\/}(2), 155--184.

\bibitem[\protect\citeauthoryear{Behrman, Sengupta, and Todd}{Behrman
  et~al.}{2005}]{Behrman2005}
Behrman, J.~R., P.~Sengupta, and P.~E. Todd (2005, oct).
\newblock {Progressing through PROGRESA: An Impact Assessment of a School
  Subsidy Experiment in Rural Mexico}.
\newblock {\em Economic Development and Cultural Change\/}~{\em 54\/}(1),
  237--275.

\bibitem[\protect\citeauthoryear{Belloni, Chen, Chernozhukov, and
  Hansen}{Belloni et~al.}{2012}]{Belloni2012}
Belloni, A., D.~Chen, V.~Chernozhukov, and C.~Hansen (2012).
\newblock {Sparse Models and Methods for Optimal Instruments with an
  Application to Eminent Domain}.
\newblock {\em Econometrica\/}~{\em 80\/}(6), 2369--2429.

\bibitem[\protect\citeauthoryear{Benhassine, Devoto, Duflo, Dupas, and
  Pouliquen}{Benhassine et~al.}{2015}]{Benhassine2015}
Benhassine, N., F.~Devoto, E.~Duflo, P.~Dupas, and V.~Pouliquen (2015).
\newblock {Turning a shove into a nudge? A "labeled cash transfer" for
  education}.
\newblock {\em American Economic Journal: Economic Policy\/}~{\em 7\/}(3),
  1--48.

\bibitem[\protect\citeauthoryear{Bhattacharya and Dupas}{Bhattacharya and
  Dupas}{2012}]{Bhattacharya2012}
Bhattacharya, D. and P.~Dupas (2012).
\newblock {Inferring welfare maximizing treatment assignment under budget
  constraints}.
\newblock {\em Journal of Econometrics\/}~{\em 167\/}(1), 168--196.

\bibitem[\protect\citeauthoryear{Cox}{Cox}{1958}]{cox58}
Cox, D.~R. (1958).
\newblock {\em Planning of Experiments}.
\newblock Wiley.

\bibitem[\protect\citeauthoryear{{De Janvry} and Sadoulet}{{De Janvry} and
  Sadoulet}{2006}]{DeJanvry2006}
{De Janvry}, A. and E.~Sadoulet (2006, mar).
\newblock {Making conditional cash transfer programs more efficient: designing
  for maximum effect of the conditionality}.
\newblock {\em The World Bank Economic Review\/}~{\em 20\/}(1), 1.

\bibitem[\protect\citeauthoryear{Dehejia, Pop-Eleches, and Samii}{Dehejia
  et~al.}{2016}]{Dehejia2016}
Dehejia, R., C.~Pop-Eleches, and C.~Samii (2016).
\newblock {From local to global: External validity in a fertility natural
  experiment}.

\bibitem[\protect\citeauthoryear{Dehejia, Pop-Eleches, and Samii}{Dehejia
  et~al.}{2017}]{Dehejia2014}
Dehejia, R., C.~Pop-Eleches, and C.~Samii (2017).
\newblock {From Local to Global: External Validity in a Fertility Natural
  Experiment}.
\newblock {\em NBER Working Paper\/}~{\em 21459}.

\bibitem[\protect\citeauthoryear{Dehejia}{Dehejia}{2003}]{Dehejia2003}
Dehejia, R.~H. (2003, jan).
\newblock {Was There a Riverside Miracle? A Hierarchical Framework for
  Evaluating Programs With Grouped Data}.
\newblock {\em Journal of Business and Economic Statistics\/}~{\em 21\/}(1),
  1--11.

\bibitem[\protect\citeauthoryear{Dehejia}{Dehejia}{2005}]{Dehejia2005}
Dehejia, R.~H. (2005, apr).
\newblock {Program evaluation as a decision problem}.
\newblock {\em Journal of Econometrics\/}~{\em 125\/}(1-2), 141--173.

\bibitem[\protect\citeauthoryear{Dehejia}{Dehejia}{2008}]{Dehejia2008}
Dehejia, R.~H. (2008).
\newblock {When is ATE enough? Risk aversion and inequality aversion in
  evaluating training programs}.
\newblock {\em Advances in Econometrics\/}~{\em 21}, 263--287.

\bibitem[\protect\citeauthoryear{DellaVigna and Pope}{DellaVigna and
  Pope}{2017}]{DellaVigna2016a}
DellaVigna, S. and D.~Pope (2017).
\newblock {Predicting Experimental Results: Who Knows What?}
\newblock {\em Journal of Political Economy, forthcoming\/}.

\bibitem[\protect\citeauthoryear{Dhaliwal, Duflo, and Glennerster}{Dhaliwal
  et~al.}{2011}]{Dhaliwal2011}
Dhaliwal, I., E.~Duflo, and R.~Glennerster (2011).
\newblock {Comparative Cost-Effectiveness Analysis to Inform Policy in
  Developing Countries: A General Framework with Applications for Education}.

\bibitem[\protect\citeauthoryear{Diebold}{Diebold}{2015}]{Diebold2015}
Diebold, F.~X. (2015).
\newblock {Comparing Predictive Accuracy, Twenty Years Later: A Personal
  Perspective on the Use and Abuse of Diebold--Mariano Tests}.
\newblock {\em Journal of Business {\&} Economic Statistics\/}~{\em 33\/}(1),
  1--1.

\bibitem[\protect\citeauthoryear{Diebold and Mariano}{Diebold and
  Mariano}{1995}]{Diebold1995}
Diebold, F.~X. and R.~S. Mariano (1995).
\newblock {Comparing Predictive Accuracy}.
\newblock {\em Journal of Business {\&} Economic Statistics\/}~{\em 13\/}(3),
  253.

\bibitem[\protect\citeauthoryear{Duflo, Hanna, and Ryan}{Duflo
  et~al.}{2012}]{Duflo2012}
Duflo, E., R.~Hanna, and S.~Ryan (2012).
\newblock {Incentives work: Getting teachers to come to school}.
\newblock {\em The American Economic Review\/}~{\em 102\/}(4), 1241--1278.

\bibitem[\protect\citeauthoryear{Garber and Phelps}{Garber and
  Phelps}{1997}]{Garber1997}
Garber, A.~M. and C.~E. Phelps (1997).
\newblock {Economic foundations of cost-effectiveness analysis}.
\newblock {\em Journal of Health Economics\/}~{\em 16\/}(1), 1--31.

\bibitem[\protect\citeauthoryear{Garcia and Saavedra}{Garcia and
  Saavedra}{2017}]{Garcia-Saavedra2017-CCT-meta}
Garcia, S. and J.~Saavedra (2017).
\newblock Educational impacts and cost-effectiveness of conditional cash
  transfer programs in developing countries: A meta-analysis.
\newblock NBER Working Paper No. 23594.

\bibitem[\protect\citeauthoryear{Gechter}{Gechter}{2016}]{Gechter2015a}
Gechter, M. (2016).
\newblock {Generalizing the Results from Social Experiments: Theory and
  Evidence from Mexico and India}.
\newblock {\em Working Paper\/}.

\bibitem[\protect\citeauthoryear{Gechter and Meager}{Gechter and
  Meager}{2018}]{Gechter2018}
Gechter, M. and R.~Meager (2018).
\newblock {Incorporating Experimental and Observational Studies in
  Meta-Analysis}.
\newblock {\em Working Paper\/}.

\bibitem[\protect\citeauthoryear{Granger and Machina}{Granger and
  Machina}{2006}]{Granger2006}
Granger, C. and M.~Machina (2006).
\newblock {Forecasting and decision theory}.
\newblock {\em Handbook of economic forecasting\/}~{\em 1\/}(05), 81--98.

\bibitem[\protect\citeauthoryear{Hansen, Lunde, and Nason}{Hansen
  et~al.}{2011}]{Hansen2011}
Hansen, P.~R., A.~Lunde, and J.~M. Nason (2011).
\newblock {The Model Confidence Set}.
\newblock {\em Econometrica\/}~{\em 79\/}(2), 453--497.

\bibitem[\protect\citeauthoryear{Hayfield and Racine}{Hayfield and
  Racine}{2008}]{Hayfield2008}
Hayfield, T. and J.~S. Racine (2008).
\newblock Nonparametric econometrics: The np package.
\newblock {\em Journal of Statistical Software\/}~{\em 27\/}(5).

\bibitem[\protect\citeauthoryear{Hirano, Imbens, and Ridder}{Hirano
  et~al.}{2003}]{Hirano2003}
Hirano, K., G.~Imbens, and G.~Ridder (2003).
\newblock {Efficient estimation of average treatment effects using the
  estimated propensity score}.
\newblock {\em Econometrica\/}~{\em 71\/}(4), 1161--1189.

\bibitem[\protect\citeauthoryear{Hirano and Porter}{Hirano and
  Porter}{2009}]{Hirano2009}
Hirano, K. and J.~R. Porter (2009).
\newblock {Asymptotics for Statistical Treatment Rules}.
\newblock {\em Econometrica\/}~{\em 77\/}(5), 1683--1701.

\bibitem[\protect\citeauthoryear{Hirano and Porter}{Hirano and
  Porter}{2019}]{Hirano2019}
Hirano, K. and J.~R. Porter (2019).
\newblock {Statistical Decision Rules in Econometrics}.
\newblock In {\em Handbook of Econometrics, Volume 7}.

\bibitem[\protect\citeauthoryear{Hotz, Imbens, and Mortimer}{Hotz
  et~al.}{2005}]{JosephHotz2005}
Hotz, V.~J., G.~Imbens, and J.~Mortimer (2005).
\newblock {Predicting the Efficacy of Future Training Programs Using Past
  Experiences at Other Locations}.
\newblock {\em Journal of Econometrics\/}~{\em 125}, 241--270.

\bibitem[\protect\citeauthoryear{Hudgens and Halloran}{Hudgens and
  Halloran}{2008}]{hudgens_halloran08}
Hudgens, M.~G. and M.~E. Halloran (2008).
\newblock Toward causal inference with interference.
\newblock {\em Journal of the American Statistical Association\/}~{\em
  103\/}(482), 832--842.

\bibitem[\protect\citeauthoryear{Imai and Strauss}{Imai and
  Strauss}{2011}]{Imai2011}
Imai, K. and A.~Strauss (2011).
\newblock {Estimation of heterogeneous treatment effects from randomized
  experiments, with application to the optimal planning of the get-out-the-vote
  campaign}.
\newblock {\em Political Analysis\/}~{\em 19\/}(1), 1--19.

\bibitem[\protect\citeauthoryear{Imbens and Wooldridge}{Imbens and
  Wooldridge}{2009}]{Imbens2009}
Imbens, G. W.~G. and J.~M. Wooldridge (2009).
\newblock {Recent developments in the econometrics of program evaluation}.
\newblock {\em Journal of Economic Literature\/}~{\em 47\/}(1), 5--86.

\bibitem[\protect\citeauthoryear{Keane, Todd, and Wolpin}{Keane
  et~al.}{2011}]{Keane2011}
Keane, M.~P., P.~E. Todd, and K.~I. Wolpin (2011).
\newblock {The Structural Estimation of Behavioral Models: Discrete Choice
  Dynamic Programming Methods and Applications}.
\newblock {\em Handbook of Labor Economics, Volume 4\/}~{\em 4\/}(2), 331--461.

\bibitem[\protect\citeauthoryear{Keane and Wolpin}{Keane and
  Wolpin}{2007}]{Keane2007a}
Keane, M.~P. and K.~I. Wolpin (2007, dec).
\newblock {Exploring the usefulness of a nonrandom holdout sample for model
  validation: welfare effects on female behavior}.
\newblock {\em International Economic Review\/}~{\em 48\/}(4), 1351--1378.

\bibitem[\protect\citeauthoryear{Kitagawa and Tetenov}{Kitagawa and
  Tetenov}{2017}]{Kitagawa2015}
Kitagawa, T. and A.~Tetenov (2017).
\newblock {Who should be treated? Empirical welfare maximization methods for
  treatment choice}.
\newblock {\em Econometrica, Forthcoming\/}.

\bibitem[\protect\citeauthoryear{Lunceford and Davidian}{Lunceford and
  Davidian}{2004}]{Lunceford2004}
Lunceford, J.~K. and M.~Davidian (2004).
\newblock Stratification and weighting via the propensity score in estimation
  of causal treatment effects: A comparative study.
\newblock {\em Statistics in Medicine\/}~{\em 23\/}(19), 2937--2960.

\bibitem[\protect\citeauthoryear{Mammen, Rothe, and Schienle}{Mammen
  et~al.}{2012}]{Mammen2012}
Mammen, E., C.~Rothe, and M.~Schienle (2012).
\newblock {Nonparametric regression with nonparametrically generated
  covariates}.
\newblock {\em Annals of Statistics\/}~{\em 40\/}(2), 1132--1170.

\bibitem[\protect\citeauthoryear{Manski}{Manski}{1996}]{Manski1996}
Manski, C.~F. (1996).
\newblock {Learning about Treatment Effects from Experiments with Random
  Assignment of Treatments}.
\newblock {\em The Journal of Human Resources\/}~{\em 31\/}(4), 709.

\bibitem[\protect\citeauthoryear{Manski}{Manski}{2004}]{Manski2004}
Manski, C.~F. (2004, jul).
\newblock {Statistical Treatment Rules for Heterogeneous Populations}.
\newblock {\em Econometrica\/}~{\em 72\/}(4), 1221--1246.

\bibitem[\protect\citeauthoryear{Manski}{Manski}{2011}]{Manski2011}
Manski, C.~F. (2011).
\newblock {Choosing Treatment Policies Under Ambiguity}.
\newblock {\em Annual Review of Economics\/}~{\em 3\/}(1), 25--49.

\bibitem[\protect\citeauthoryear{Marschak and Miyasawa}{Marschak and
  Miyasawa}{1968}]{Marschak1968}
Marschak, J. and K.~Miyasawa (1968).
\newblock {Economic Comparability of Information Systems}.
\newblock {\em Internal Economic Review\/}~{\em 9\/}(2), 137--174.

\bibitem[\protect\citeauthoryear{Meager}{Meager}{2016}]{Meager2015}
Meager, R. (2016).
\newblock {Understanding the Impact of Microcredit Expansions: A Bayesian
  Hierarchical Analysis of 7 Randomised Experiments}.

\bibitem[\protect\citeauthoryear{Montenegro and Patrinos}{Montenegro and
  Patrinos}{2014}]{Montenegro2014}
Montenegro, C.~E. and H.~A. Patrinos (2014).
\newblock {Comparable Estimates of Returns to Schooling Around the World}.
\newblock {\em World Bank Policy Research Working Paper 7020\/}.

\bibitem[\protect\citeauthoryear{Newey and McFadden}{Newey and
  McFadden}{1994}]{Newey1994}
Newey, W. and D.~McFadden (1994).
\newblock {Large sample estimation and hypothesis testing}.
\newblock {\em Handbook of econometrics\/}~{\em 4}, 2111--2245.

\bibitem[\protect\citeauthoryear{Parker and Vogl}{Parker and
  Vogl}{2018}]{Parker2018}
Parker, S.~W. and T.~Vogl (2018).
\newblock {Do conditional cash transfers improve economic outcomes in the next
  generation? Evidence from mexico}.

\bibitem[\protect\citeauthoryear{Pesaran and Skouras}{Pesaran and
  Skouras}{2002}]{Pesaran2002}
Pesaran, M. and S.~Skouras (2002).
\newblock {Decision-Based Methods for Forecast Evaluation}.
\newblock In M.~Clements and D.~Hendry (Eds.), {\em A Companion to Economic
  Forecasting}. Oxford: Blackwell.

\bibitem[\protect\citeauthoryear{Pritchett and Sandefur}{Pritchett and
  Sandefur}{2013}]{Pritchett2013}
Pritchett, L. and J.~Sandefur (2013).
\newblock {Context Matters for Size: Why External Validity Claims and
  Development Practice Don't Mix}.
\newblock {\em Journal of Globalization and Development\/}~{\em 4\/}(2),
  161--197.

\bibitem[\protect\citeauthoryear{Rubin}{Rubin}{1980}]{rubin80_randomization_analysis}
Rubin, D.~B. (1980).
\newblock Randomization analysis of experimental data: The {Fisher}
  randomization test comment.
\newblock {\em Journal of the American Statistical Association\/}~{\em
  75\/}(371), 591--593.

\bibitem[\protect\citeauthoryear{Schorfheide and Wolpin}{Schorfheide and
  Wolpin}{2012}]{Schorfheide2012}
Schorfheide, F. and K.~I. Wolpin (2012, may).
\newblock {On the Use of Holdout Samples for Model Selection}.
\newblock {\em American Economic Review\/}~{\em 102\/}(3), 477--481.

\bibitem[\protect\citeauthoryear{Schorfheide and Wolpin}{Schorfheide and
  Wolpin}{2016}]{Schorfheide2016}
Schorfheide, F. and K.~I. Wolpin (2016).
\newblock {To Hold Out or Not To Hold Out}.
\newblock {\em Research in Economics, forthcoming\/}.

\bibitem[\protect\citeauthoryear{Schultz}{Schultz}{2004}]{Schultz2004}
Schultz, T.~P. (2004, jun).
\newblock {School subsidies for the poor: evaluating the Mexican Progresa
  poverty program}.
\newblock {\em Journal of Development Economics\/}~{\em 74\/}(1), 199--250.

\bibitem[\protect\citeauthoryear{Tetenov}{Tetenov}{2012}]{Tetenov2012}
Tetenov, A. (2012).
\newblock {Statistical treatment choice based on asymmetric minimax regret
  criteria}.
\newblock {\em Journal of Econometrics\/}~{\em 166\/}(1), 157--165.

\bibitem[\protect\citeauthoryear{Todd and Wolpin}{Todd and
  Wolpin}{2006}]{Todd2006}
Todd, P.~E. and K.~I. Wolpin (2006).
\newblock {Assessing the impact of a school subsidy program in Mexico: Using a
  social experiment to validate a dynamic behavioral model of child schooling
  and fertility}.
\newblock {\em The American Economic Review\/}~{\em 96\/}(5), 1384--1417.

\bibitem[\protect\citeauthoryear{Todd and Wolpin}{Todd and
  Wolpin}{2008}]{Todd2010a}
Todd, P.~E. and K.~I. Wolpin (2008).
\newblock {Ex ante evaluation of social programs}.
\newblock {\em Annales d'Economie et de Statistique\/}~(91/92), 263--291.

\bibitem[\protect\citeauthoryear{Vivalt}{Vivalt}{2016}]{Vivalt2014}
Vivalt, E. (2016).
\newblock {How Much Can We Generalize From Impact Evaluations?}
\newblock {\em Mimeo\/}.

\bibitem[\protect\citeauthoryear{Vuong}{Vuong}{1989}]{Vuong1989}
Vuong, Q. (1989).
\newblock {Likelihood Ratio Tests for Model Selection and Non-Nested
  Hypotheses}.
\newblock {\em Econometrica\/}~{\em 57\/}(2), 307--333.

\bibitem[\protect\citeauthoryear{West}{West}{1996}]{West1996}
West, K. (1996).
\newblock {Asymptotic inference about predictive ability}.
\newblock {\em Econometrica\/}~{\em 64\/}(5), 1067--1084.

\bibitem[\protect\citeauthoryear{White}{White}{2000}]{White2003}
White, H. (2000).
\newblock {A reality check for data snooping}.
\newblock {\em Econometrica\/}~{\em 68\/}(5), 1097--1126.

\bibitem[\protect\citeauthoryear{Wolpin}{Wolpin}{2007}]{Wolpin2007}
Wolpin, K. (2007, may).
\newblock {Ex Ante Policy Evaluation, Structural Estimation, and Model
  Selection}.
\newblock {\em The American economic review\/}~{\em 97\/}(2), 48--52.

\bibitem[\protect\citeauthoryear{Wolpin}{Wolpin}{2013}]{Wolpin2013}
Wolpin, K.~I. (2013).
\newblock {\em {The Limits of Inference Without Theory}}.
\newblock The MIT Press.

\end{thebibliography}

\newpage

\begin{figure}
	\caption{SPS model: original density for boys}
	\label{fig:sps_original_density}
	\includegraphics[width=\hsize]{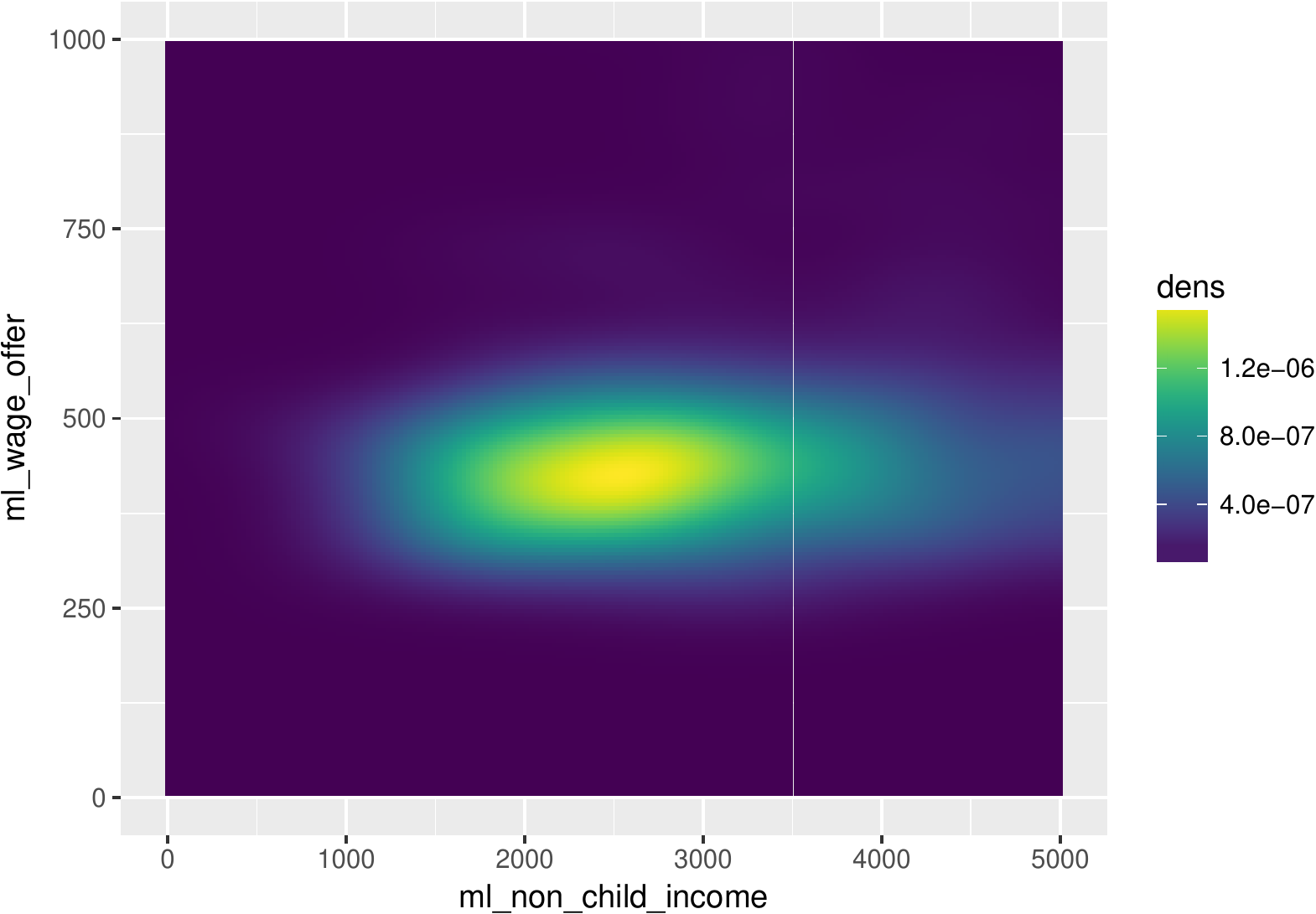}
\end{figure}

\begin{figure}
	\caption{SPS model: counterfactual density for boys}
	\label{fig:sps_counterfactual_density}
	\includegraphics[width=\hsize]{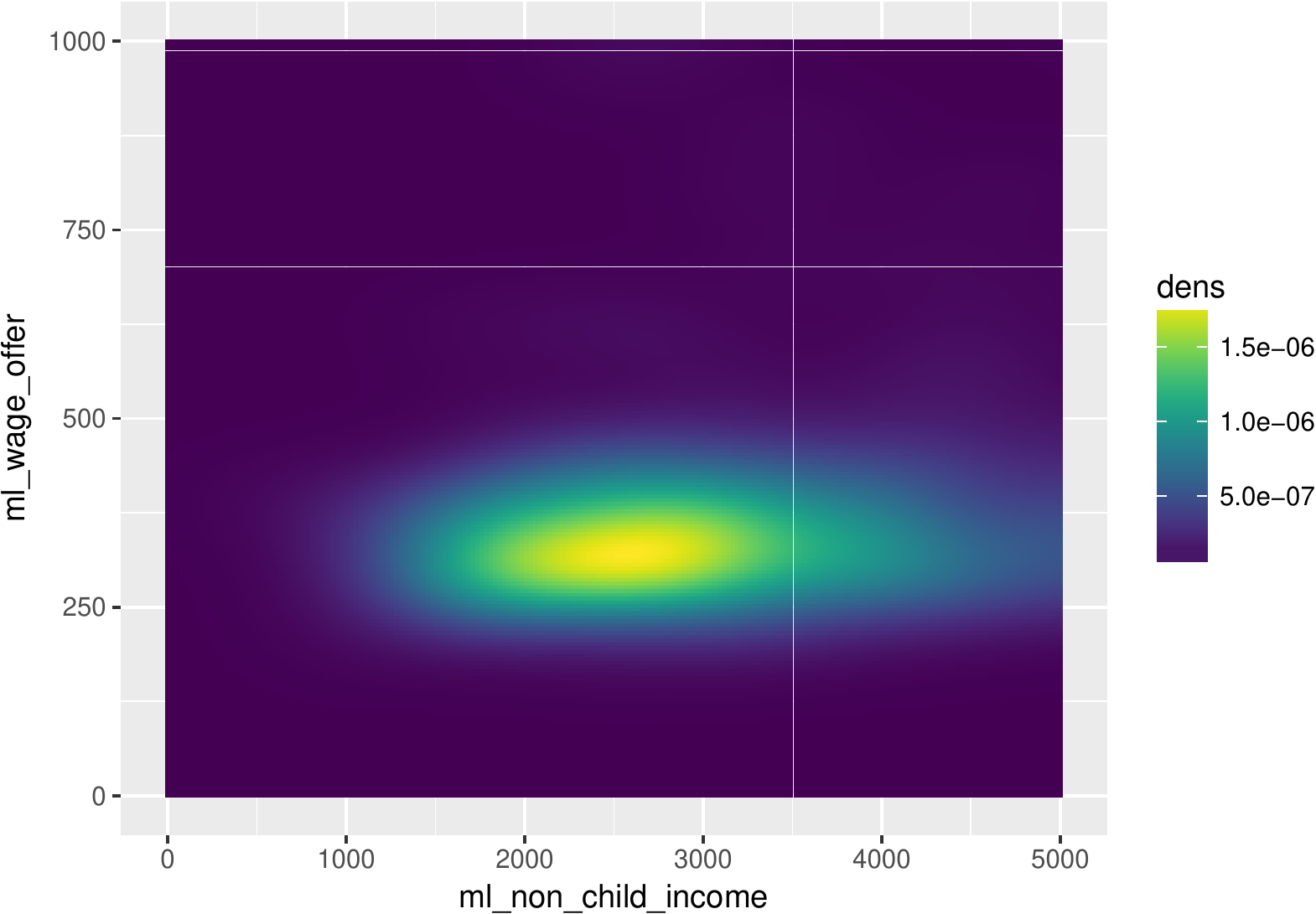}
\end{figure}

\begin{figure}
	\caption{SPS: regression function}
	\label{fig:sps_bivariate_regression}
	\includegraphics[width=\hsize]{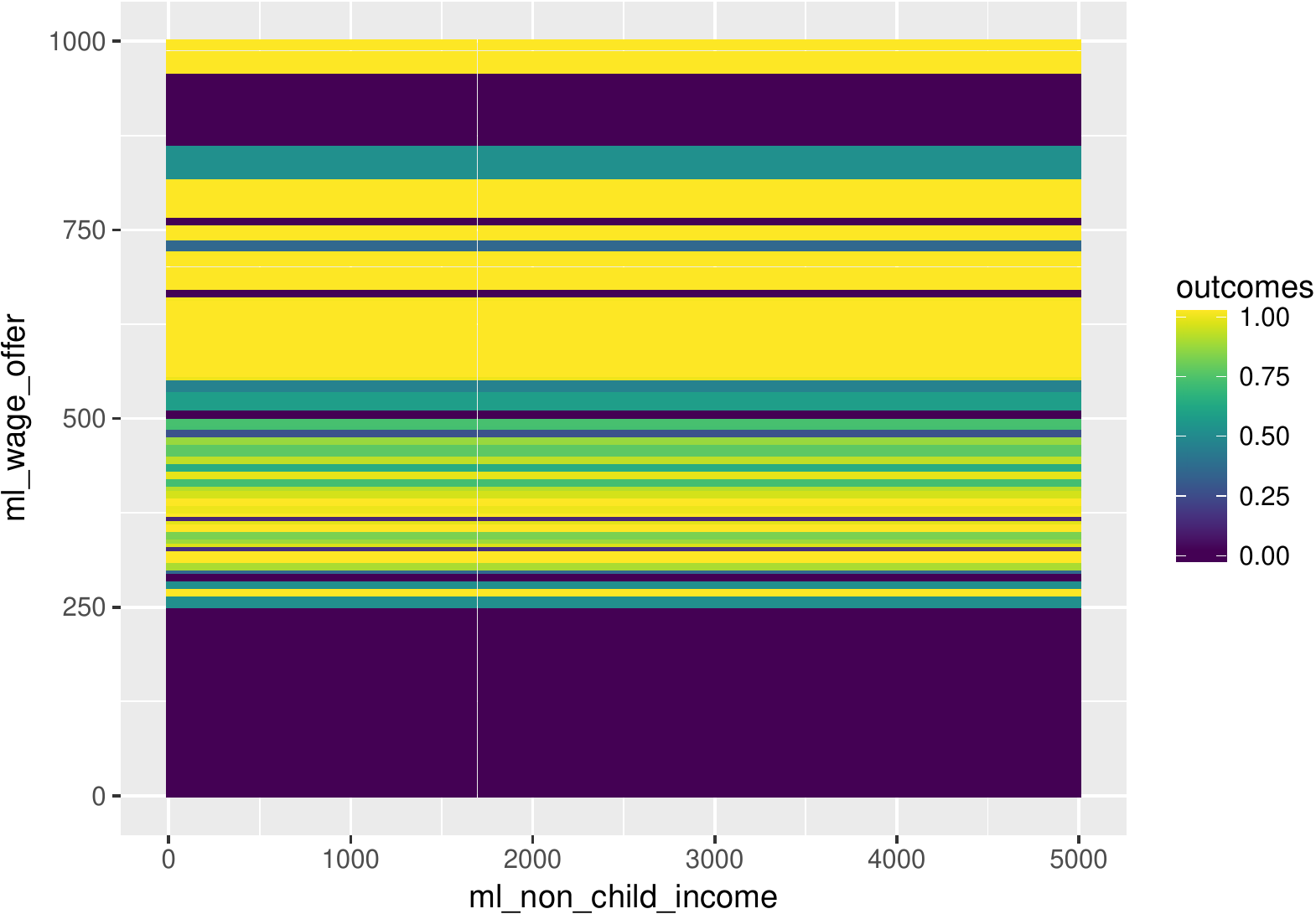}
\end{figure}

\begin{figure}
	\caption{DPS: in-sample fit by age}
	\label{fig:dps_fit_by_age}
	\includegraphics[width=\hsize]{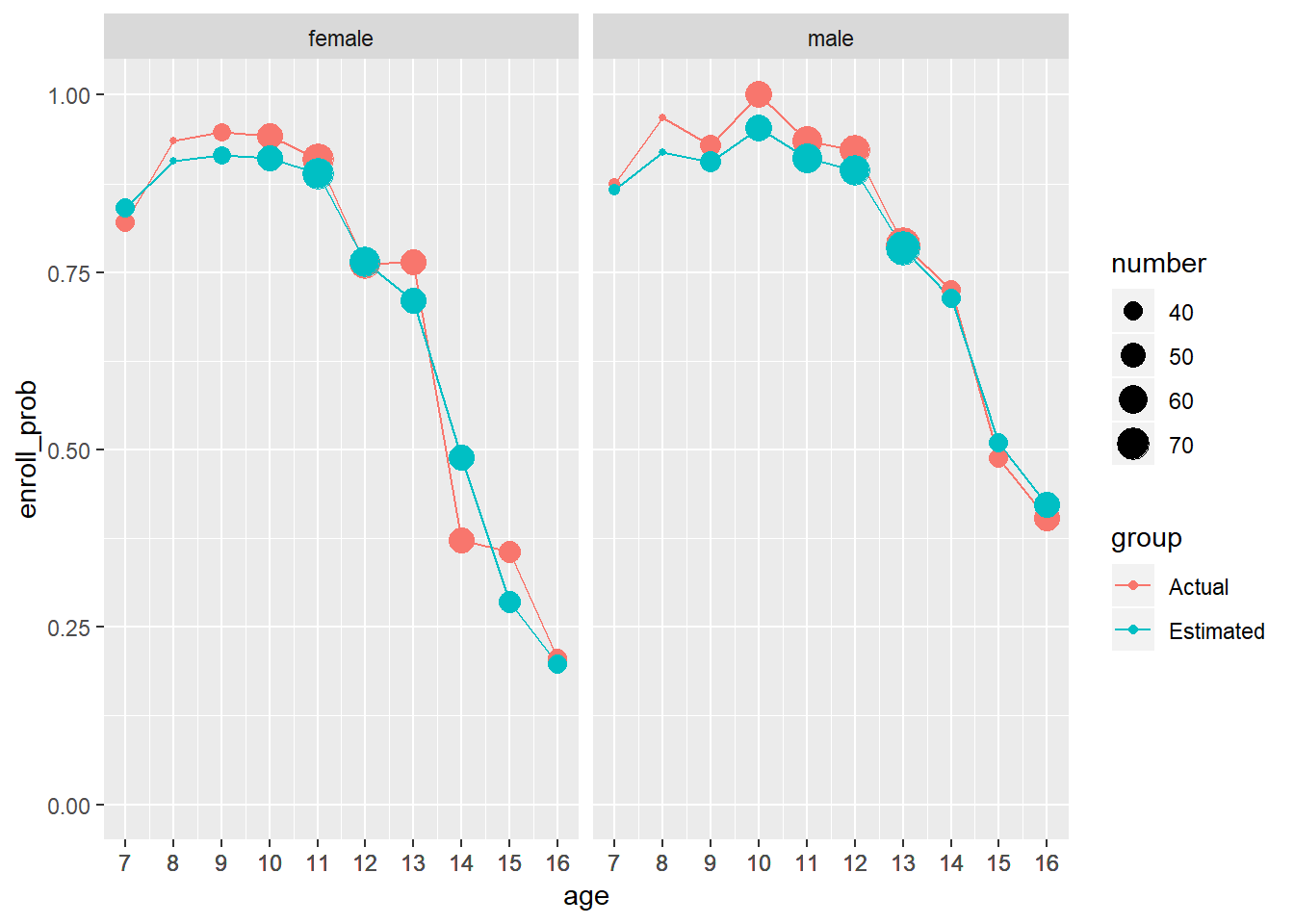}
\end{figure}

\begin{figure}
	\caption{DPS: in-sample fit by years of education completed}
	\label{fig:dps_fit_by_ed}
	\includegraphics[width=\hsize]{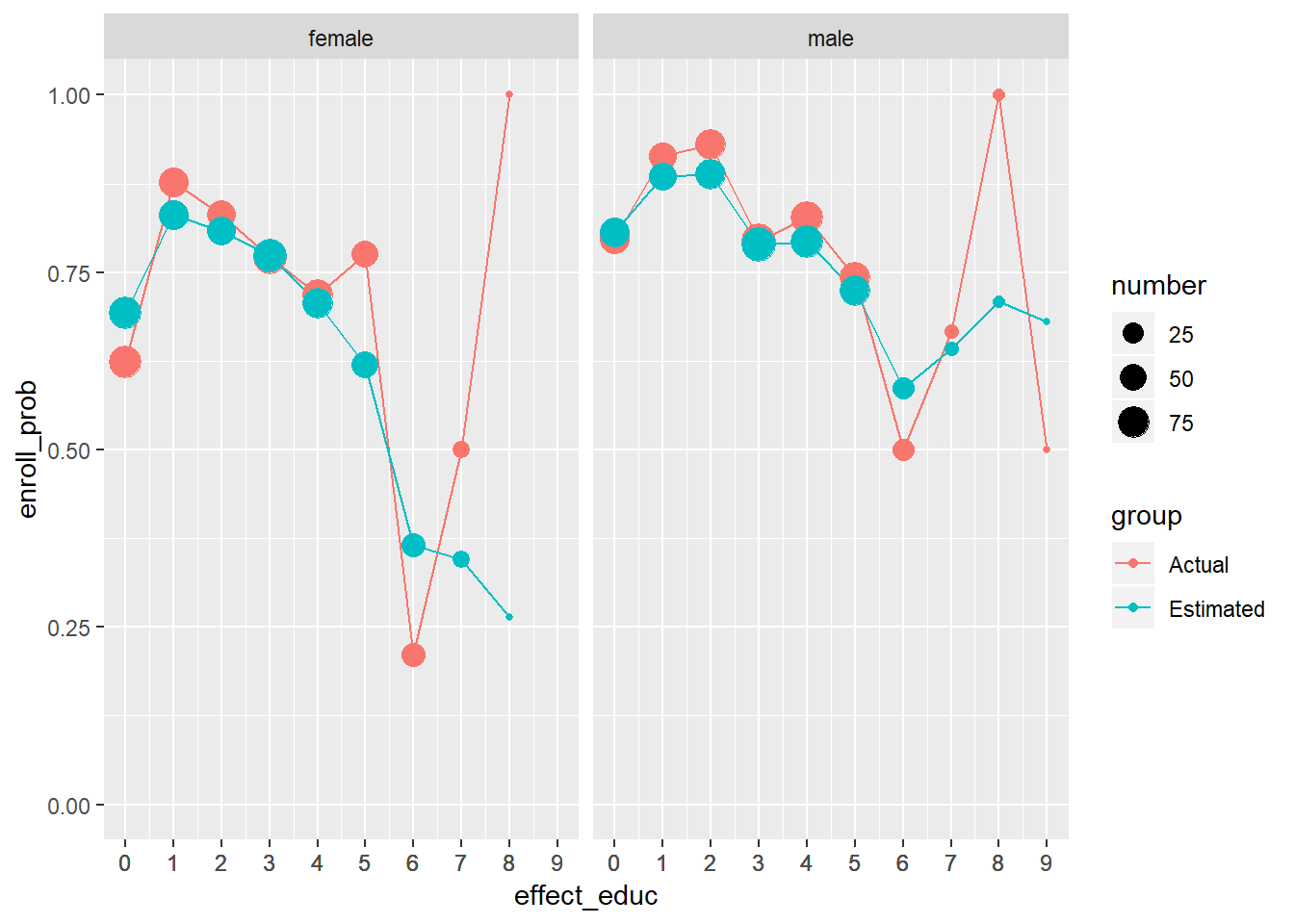}
\end{figure}

\begin{figure}
	\caption{DPS predicted treatment effect, no cost-effectiveness adjustment}
	\label{fig:dps_predicted_te_dist}
	\includegraphics[width=\hsize]{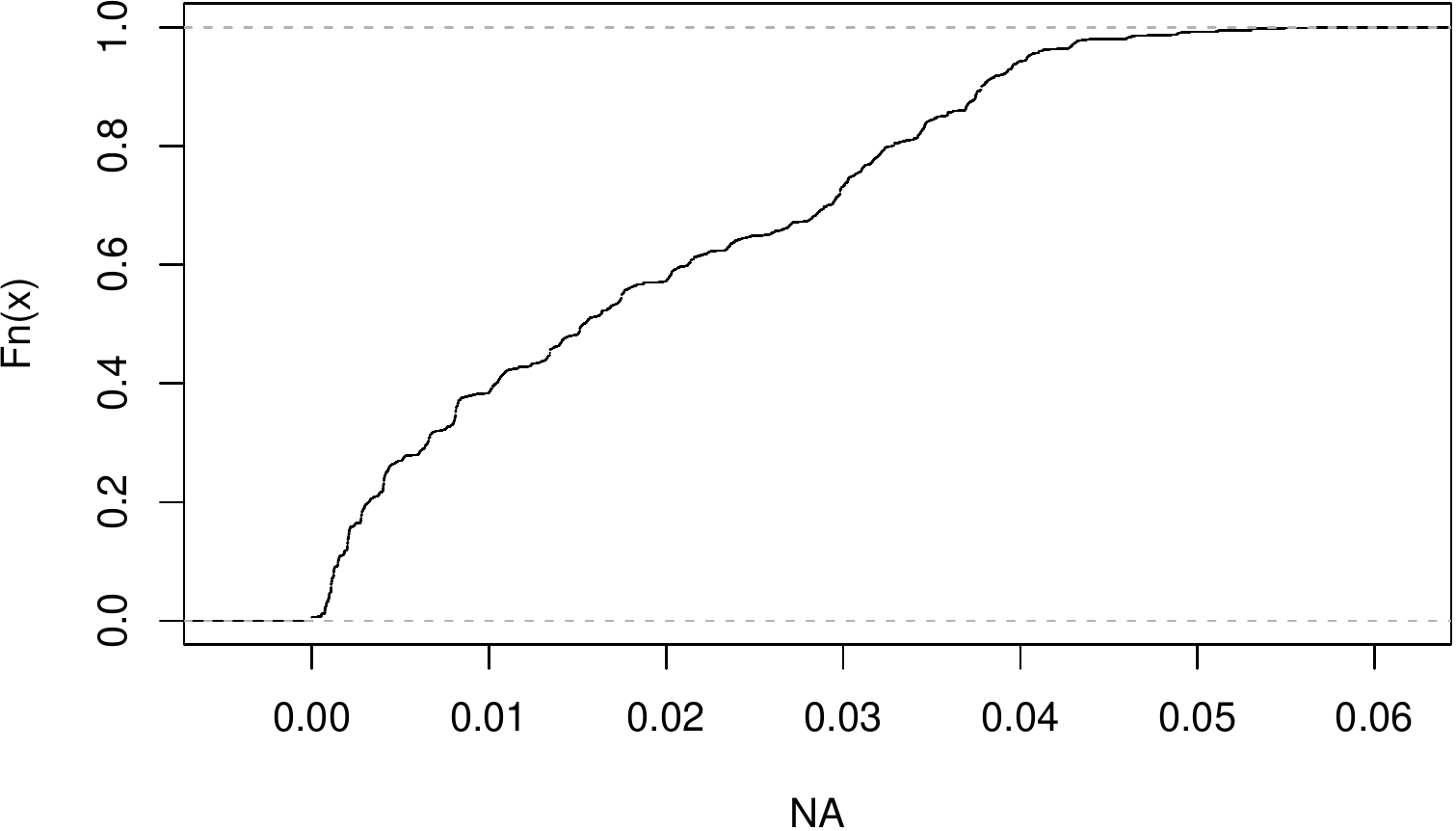}
\end{figure}


\begin{center}
	\begin{table}[ht]
\centering
\begin{tabular}{rrrrrrr}
  \hline
Age & Male & Treated & Control & CATE & PEN\_ub & PNE\_ub \\ 
  \hline
6.00 & 1.00 & 0.93 & 0.91 & 0.02 & 0.09 & 0.07 \\ 
  7.00 & 1.00 & 0.96 & 0.96 & -0.00 & 0.04 & 0.04 \\ 
  8.00 & 1.00 & 0.97 & 0.96 & 0.01 & 0.04 & 0.03 \\ 
  9.00 & 1.00 & 0.98 & 0.96 & 0.02 & 0.04 & 0.02 \\ 
  10.00 & 1.00 & 0.97 & 0.96 & 0.01 & 0.04 & 0.03 \\ 
  11.00 & 1.00 & 0.96 & 0.93 & 0.03 & 0.07 & 0.04 \\ 
  12.00 & 1.00 & 0.91 & 0.86 & 0.05 & 0.14 & 0.09 \\ 
  13.00 & 1.00 & 0.83 & 0.76 & 0.07 & 0.24 & 0.17 \\ 
  14.00 & 1.00 & 0.74 & 0.60 & 0.14 & 0.40 & 0.26 \\ 
  15.00 & 1.00 & 0.53 & 0.46 & 0.07 & 0.53 & 0.46 \\ 
  16.00 & 1.00 & 0.35 & 0.32 & 0.03 & 0.35 & 0.32 \\ 
  6.00 & 0.00 & 0.90 & 0.89 & 0.01 & 0.11 & 0.10 \\ 
  7.00 & 0.00 & 0.95 & 0.95 & 0.00 & 0.05 & 0.05 \\ 
  8.00 & 0.00 & 0.96 & 0.96 & 0.00 & 0.04 & 0.04 \\ 
  9.00 & 0.00 & 0.97 & 0.95 & 0.01 & 0.05 & 0.03 \\ 
  10.00 & 0.00 & 0.97 & 0.96 & 0.00 & 0.04 & 0.03 \\ 
  11.00 & 0.00 & 0.94 & 0.92 & 0.03 & 0.08 & 0.06 \\ 
  12.00 & 0.00 & 0.88 & 0.79 & 0.09 & 0.21 & 0.12 \\ 
  13.00 & 0.00 & 0.74 & 0.68 & 0.05 & 0.32 & 0.26 \\ 
  14.00 & 0.00 & 0.65 & 0.52 & 0.13 & 0.48 & 0.35 \\ 
  15.00 & 0.00 & 0.41 & 0.35 & 0.06 & 0.41 & 0.35 \\ 
  16.00 & 0.00 & 0.33 & 0.25 & 0.07 & 0.33 & 0.25 \\ 
   \hline
\end{tabular}
\caption{Age-sex Outcomes for Mexico} 
\label{tab:age-sex-cates-mex}
\end{table}

	\begin{table}[ht]
\centering
\begin{tabular}{rrrrr}
  \hline
Age & Boys treated & Boys control & Girls treated & Girls control \\ 
  \hline
6.00 & 0.84 & 0.91 & 0.82 & 0.89 \\ 
  7.00 & 0.87 & 0.96 & 0.87 & 0.95 \\ 
  8.00 & 0.85 & 0.96 & 0.85 & 0.96 \\ 
  9.00 & 0.86 & 0.96 & 0.85 & 0.95 \\ 
  10.00 & 0.83 & 0.96 & 0.82 & 0.96 \\ 
  11.00 & 0.82 & 0.93 & 0.80 & 0.92 \\ 
  12.00 & 0.77 & 0.86 & 0.75 & 0.79 \\ 
  13.00 & 0.70 & 0.76 & 0.63 & 0.68 \\ 
  14.00 & 0.63 & 0.60 & 0.56 & 0.52 \\ 
  15.00 & 0.45 & 0.46 & 0.35 & 0.35 \\ 
  16.00 & 0.30 & 0.32 & 0.28 & 0.25 \\ 
   \hline
\end{tabular}
\caption{Adjusted Age-sex Outcomes for Mexico} 
\label{tab:age-sex-adj-cates-mex}
\end{table}

\begin{table}[ht]
\centering
\begin{tabular}{rrr}
  \hline
 & w/ CE & w/o CE \\ 
  \hline
Share treated (age-sex) & 0.126 & 0.966 \\ 
  Share treated (all) & 1.000 & 1.000 \\ 
  Enrollment difference & -0.035 & -0.005 \\ 
  SE enroll. diff. & 0.011 & 0.004 \\ 
  Welfare difference & 0.065 & -0.002 \\ 
  SE welfare diff. & 0.010 & 0.004 \\ 
   \hline
\end{tabular}
\caption{Welfare comparison for age-sex extrapolation vs. treat all} 
\label{tab:agesex-all-welfare}
\end{table}
 
\end{center}


\begin{center}
\begin{table}[ht]
\centering
\begin{tabular}{lrr}
  \hline
Variable & Importance Y1 & Importance Y0 \\ 
  \hline
ml\_base\_enrolled & 0.21 & 0.19 \\ 
  ml\_age & 0.33 & 0.32 \\ 
  ml\_n\_child & 0.02 & 0.01 \\ 
  ml\_male & 0.01 & 0.01 \\ 
  ml\_hh\_head\_male & 0.00 & 0.00 \\ 
  ml\_single\_parent & 0.00 & 0.00 \\ 
  ml\_al\_father & 0.00 & 0.01 \\ 
  ml\_al\_mother & 0.00 & 0.00 \\ 
  ml\_lw\_father & 0.00 & 0.00 \\ 
  ml\_lw\_mother & 0.00 & 0.00 \\ 
  ml\_literacy & 0.03 & 0.04 \\ 
  ml\_n\_total & 0.01 & 0.01 \\ 
  ml\_yrs\_educ & 0.16 & 0.18 \\ 
  ml\_hh\_head\_edu & 0.07 & 0.07 \\ 
  ml\_father\_educ & 0.05 & 0.04 \\ 
  ml\_mother\_educ & 0.07 & 0.08 \\ 
  ml\_base\_enrolled\_mi & 0.01 & 0.01 \\ 
  ml\_male\_mi & 0.00 & 0.00 \\ 
  ml\_al\_father\_mi & 0.00 & 0.00 \\ 
  ml\_al\_mother\_mi & 0.00 & 0.00 \\ 
  ml\_literacy\_mi & 0.01 & 0.01 \\ 
  ml\_yrs\_educ\_mi & 0.01 & 0.01 \\ 
  ml\_father\_educ\_mi & 0.01 & 0.00 \\ 
  ml\_mother\_educ\_mi & 0.00 & 0.00 \\ 
  ml\_hh\_monthly\_consump & 0.00 & 0.00 \\ 
  ml\_n\_child\_mi & 0.00 & 0.00 \\ 
  ml\_lw\_father\_mi & 0.00 & 0.00 \\ 
  ml\_hh\_monthly\_consump\_mi & 0.00 & 0.00 \\ 
   \hline
\end{tabular}
\caption{GRF variable importance: no cost-effectiveness adjustment} 
\label{tab:mexico-grfpo-noce-varimp}
\end{table}
 
\begin{table}[ht]
\centering
\begin{tabular}{lrr}
  \hline
Variable & Importance Y1 & Importance Y0 \\ 
  \hline
ml\_base\_enrolled & 0.19 & 0.19 \\ 
  ml\_age & 0.28 & 0.27 \\ 
  ml\_n\_child & 0.02 & 0.01 \\ 
  ml\_male & 0.01 & 0.01 \\ 
  ml\_hh\_head\_male & 0.01 & 0.00 \\ 
  ml\_single\_parent & 0.00 & 0.00 \\ 
  ml\_al\_father & 0.01 & 0.01 \\ 
  ml\_al\_mother & 0.03 & 0.03 \\ 
  ml\_lw\_father & 0.00 & 0.00 \\ 
  ml\_lw\_mother & 0.01 & 0.01 \\ 
  ml\_literacy & 0.03 & 0.04 \\ 
  ml\_n\_total & 0.01 & 0.01 \\ 
  ml\_yrs\_educ & 0.16 & 0.15 \\ 
  ml\_hh\_head\_edu & 0.08 & 0.08 \\ 
  ml\_father\_educ & 0.05 & 0.04 \\ 
  ml\_mother\_educ & 0.06 & 0.08 \\ 
  ml\_base\_enrolled\_mi & 0.01 & 0.01 \\ 
  ml\_male\_mi & 0.01 & 0.00 \\ 
  ml\_al\_father\_mi & 0.00 & 0.00 \\ 
  ml\_al\_mother\_mi & 0.00 & 0.00 \\ 
  ml\_literacy\_mi & 0.01 & 0.02 \\ 
  ml\_yrs\_educ\_mi & 0.01 & 0.02 \\ 
  ml\_father\_educ\_mi & 0.00 & 0.00 \\ 
  ml\_mother\_educ\_mi & 0.01 & 0.00 \\ 
   \hline
\end{tabular}
\caption{GRF variable importance: Y1 adjusted for cost-effectiveness} 
\label{tab:mexico-grfpo-ce-varimp}
\end{table}

\begin{table}[ht]
\centering
\begin{tabular}{rrrrrr}
  \hline
Age & Male & Avg. TE & Min. TE & Max TE & SD TE \\ 
  \hline
6.00 & 1.00 & -0.06 & -0.20 & 0.08 & 0.06 \\ 
  7.00 & 1.00 & -0.05 & -0.18 & 0.12 & 0.06 \\ 
  8.00 & 1.00 & 0.03 & -0.07 & 0.12 & 0.04 \\ 
  9.00 & 1.00 & 0.06 & -0.04 & 0.16 & 0.04 \\ 
  10.00 & 1.00 & 0.06 & -0.10 & 0.16 & 0.03 \\ 
  11.00 & 1.00 & 0.04 & -0.07 & 0.15 & 0.03 \\ 
  12.00 & 1.00 & 0.05 & -0.14 & 0.18 & 0.05 \\ 
  13.00 & 1.00 & 0.07 & -0.19 & 0.19 & 0.09 \\ 
  14.00 & 1.00 & 0.15 & -0.11 & 0.41 & 0.08 \\ 
  15.00 & 1.00 & 0.20 & -0.16 & 0.45 & 0.11 \\ 
  16.00 & 1.00 & 0.09 & -0.26 & 0.32 & 0.14 \\ 
   \hline
\end{tabular}
\caption{GRF CATE predictions for Morocco: Boys, no cost-effectiveness adjustment} 
\label{tab:morocco-grfpo-noce-catesum-boys}
\end{table}

\begin{table}[ht]
\centering
\begin{tabular}{rrrrrr}
  \hline
Age & Male & Avg. TE & Min. TE & Max TE & SD TE \\ 
  \hline
6.00 & 1.00 & -0.04 & -0.21 & 0.14 & 0.05 \\ 
  7.00 & 1.00 & -0.02 & -0.14 & 0.12 & 0.05 \\ 
  8.00 & 1.00 & 0.01 & -0.09 & 0.17 & 0.05 \\ 
  9.00 & 1.00 & 0.03 & -0.12 & 0.17 & 0.06 \\ 
  10.00 & 1.00 & 0.03 & -0.09 & 0.14 & 0.04 \\ 
  11.00 & 1.00 & 0.02 & -0.08 & 0.15 & 0.03 \\ 
  12.00 & 1.00 & 0.04 & -0.14 & 0.16 & 0.04 \\ 
  13.00 & 1.00 & 0.05 & -0.18 & 0.19 & 0.09 \\ 
  14.00 & 1.00 & 0.14 & -0.08 & 0.41 & 0.09 \\ 
  15.00 & 1.00 & 0.14 & -0.12 & 0.39 & 0.13 \\ 
  16.00 & 1.00 & 0.06 & -0.18 & 0.38 & 0.11 \\ 
   \hline
\end{tabular}
\caption{GRF CATE predictions for Morocco: Girls, no cost-effectiveness adjustment} 
\label{tab:morocco-grfpo-noce-catesum-girls}
\end{table}

\begin{table}[ht]
\centering
\begin{tabular}{rrrrrr}
  \hline
Age & Male & Avg. Adj. TE & Min. Adj. TE & Max Adj. TE & SD Adj. TE \\ 
  \hline
6.00 & 1.00 & -0.10 & -0.22 & 0.05 & 0.06 \\ 
  7.00 & 1.00 & -0.08 & -0.18 & 0.06 & 0.05 \\ 
  8.00 & 1.00 & -0.03 & -0.11 & 0.09 & 0.03 \\ 
  9.00 & 1.00 & -0.00 & -0.07 & 0.12 & 0.04 \\ 
  10.00 & 1.00 & -0.02 & -0.13 & 0.09 & 0.03 \\ 
  11.00 & 1.00 & -0.03 & -0.12 & 0.08 & 0.03 \\ 
  12.00 & 1.00 & -0.01 & -0.16 & 0.11 & 0.04 \\ 
  13.00 & 1.00 & 0.01 & -0.24 & 0.11 & 0.07 \\ 
  14.00 & 1.00 & 0.07 & -0.11 & 0.26 & 0.06 \\ 
  15.00 & 1.00 & 0.14 & -0.17 & 0.33 & 0.10 \\ 
  16.00 & 1.00 & 0.05 & -0.23 & 0.23 & 0.11 \\ 
   \hline
\end{tabular}
\caption{GRF CATE predictions for Morocco: Boys, adjusted for cost-effectiveness} 
\label{tab:morocco-grfpo-ce-catesum-boys}
\end{table}

\begin{table}[ht]
\centering
\begin{tabular}{rrrrrr}
  \hline
Age & Male & Avg. Adj. TE & Min. Adj. TE & Max Adj. TE & SD Adj. TE \\ 
  \hline
6.00 & 1.00 & -0.08 & -0.24 & 0.06 & 0.04 \\ 
  7.00 & 1.00 & -0.06 & -0.16 & 0.05 & 0.04 \\ 
  8.00 & 1.00 & -0.05 & -0.17 & 0.09 & 0.05 \\ 
  9.00 & 1.00 & -0.04 & -0.14 & 0.11 & 0.05 \\ 
  10.00 & 1.00 & -0.05 & -0.14 & 0.04 & 0.03 \\ 
  11.00 & 1.00 & -0.05 & -0.13 & 0.12 & 0.02 \\ 
  12.00 & 1.00 & -0.03 & -0.16 & 0.05 & 0.04 \\ 
  13.00 & 1.00 & -0.01 & -0.22 & 0.10 & 0.07 \\ 
  14.00 & 1.00 & 0.07 & -0.08 & 0.28 & 0.06 \\ 
  15.00 & 1.00 & 0.09 & -0.12 & 0.31 & 0.11 \\ 
  16.00 & 1.00 & 0.02 & -0.21 & 0.27 & 0.09 \\ 
   \hline
\end{tabular}
\caption{GRF CATE predictions for Morocco: Girls, adjusted for cost-effectiveness} 
\label{tab:morocco-grfpo-ce-catesum-girls}
\end{table}

\begin{table}[ht]
\centering
\begin{tabular}{rrrrr}
  \hline
Age & Girls share & Girls treat rate & Boys share & Boys treat rate \\ 
  \hline
6.00 & 0.05 & 0.18 & 0.06 & 0.17 \\ 
  7.00 & 0.07 & 0.28 & 0.07 & 0.23 \\ 
  8.00 & 0.09 & 0.55 & 0.08 & 0.81 \\ 
  9.00 & 0.08 & 0.73 & 0.08 & 0.94 \\ 
  10.00 & 0.10 & 0.86 & 0.12 & 0.94 \\ 
  11.00 & 0.12 & 0.78 & 0.12 & 0.88 \\ 
  12.00 & 0.13 & 0.86 & 0.13 & 0.91 \\ 
  13.00 & 0.12 & 0.71 & 0.12 & 0.83 \\ 
  14.00 & 0.09 & 0.95 & 0.10 & 0.98 \\ 
  15.00 & 0.09 & 0.91 & 0.08 & 0.94 \\ 
  16.00 & 0.06 & 0.80 & 0.05 & 0.79 \\ 
   \hline
\end{tabular}
\caption{GRF treatment rates for Morocco, no cost-effectiveness adjustment} 
\label{tab:morocco-grfpo-noce-treat-rate}
\end{table}

\begin{table}[ht]
\centering
\begin{tabular}{rrrrr}
  \hline
Age & Girls share & Girls treat rate & Boys share & Boys treat rate \\ 
  \hline
6.00 & 0.05 & 0.05 & 0.06 & 0.04 \\ 
  7.00 & 0.07 & 0.09 & 0.07 & 0.04 \\ 
  8.00 & 0.09 & 0.14 & 0.08 & 0.20 \\ 
  9.00 & 0.08 & 0.27 & 0.08 & 0.41 \\ 
  10.00 & 0.10 & 0.09 & 0.12 & 0.27 \\ 
  11.00 & 0.12 & 0.03 & 0.12 & 0.12 \\ 
  12.00 & 0.13 & 0.19 & 0.13 & 0.39 \\ 
  13.00 & 0.12 & 0.51 & 0.12 & 0.62 \\ 
  14.00 & 0.09 & 0.89 & 0.10 & 0.93 \\ 
  15.00 & 0.09 & 0.79 & 0.08 & 0.91 \\ 
  16.00 & 0.06 & 0.58 & 0.05 & 0.73 \\ 
   \hline
\end{tabular}
\caption{GRF treatment rates for Morocco, adjusted for cost-effectiveness} 
\label{tab:morocco-grfpo-ce-treat-rate}
\end{table}

\begin{table}[ht]
\centering
\begin{tabular}{rrr}
  \hline
 & w/ CE & w/o CE \\ 
  \hline
Share treated (GRF) & 0.377 & 0.772 \\ 
  Share treated (age-sex) & 0.126 & 0.966 \\ 
  Enrollment difference & 0.011 & -0.002 \\ 
  SE enroll. diff. & 0.010 & 0.008 \\ 
  Welfare difference & -0.018 & 0.014 \\ 
  SE welfare diff. & 0.009 & 0.008 \\ 
   \hline
\end{tabular}
\caption{Welfare comparison for GRF vs. age-sex extrapolation} 
\label{tab:grf-agesex-welfare}
\end{table}

\begin{table}[ht]
\centering
\begin{tabular}{rr}
  \hline
 & w/ CE \\ 
  \hline
Share treated (SPS) & 0.447 \\ 
  Share treated (age-sex) & 0.126 \\ 
  Enrollment difference & -0.018 \\ 
  SE enroll. diff. & 0.010 \\ 
  Welfare difference & -0.046 \\ 
  SE welfare diff. & 0.009 \\ 
   \hline
\end{tabular}
\caption{Welfare comparison for SPS vs. age-sex extrapolation} 
\label{tab:sps-agesex-welfare}
\end{table}

\begin{table}[ht]
\centering
\begin{tabular}{rr}
  \hline
 & w/ CE \\ 
  \hline
Share treated (DPS) & 0.001 \\ 
  Share treated (age-sex) & 0.151 \\ 
  Enrollment difference & -0.017 \\ 
  SE enroll. diff. & 0.008 \\ 
  Welfare difference & -0.004 \\ 
  SE welfare diff. & 0.008 \\ 
   \hline
\end{tabular}
\caption{Welfare comparison for DPS vs. age-sex extrapolation} 
\label{tab:dps-agesex-welfare}
\end{table}

\end{center}

\end{document}